\documentclass[12pt]{article}
\usepackage{amssymb}
\usepackage{amsthm}
\usepackage{multirow}
\usepackage{amsmath}
\usepackage{booktabs}
\usepackage{tabularx}
\usepackage{xcolor}
\usepackage{mathptmx}
\usepackage{graphicx}
\usepackage{times}
\usepackage{hyperref}
\usepackage{mathtools}

\newtheorem{theorem}{Theorem}
\newtheorem{corollary}{Corollary}
\newtheorem{observation}{Observation}
\newtheorem{definition}{Definition}
\newtheorem{lemma}{Lemma}
\newtheorem{example}{Example}

\newcommand{\elec}{\mathcal{E}} 
\newcommand{\bigo}[1]{O(#1)}
\newcommand{\bigos}[1]{O^*(#1)}

\newcommand{\yes}{\mbox{Yes}}
\newcommand{\no}{\mbox{No}}
\newcommand{\yesins}{{\yes}-instance}
\newcommand{\noins}{{\no}-instance}
\newcommand{\poly}{\textsf{P}}
\newcommand{\bpoly}{{\textbf{\textsf{P}}}}
\newcommand{\np}{\textsf{\mbox{NP}}}
\newcommand{\nph}{{\np}-{hard}}

\newcommand{\conphns}{co{\np}-{hardness}}
\newcommand{\conphshort}{co{\np}-{h}}
\newcommand{\nphns}{\np-hardness}
\newcommand{\npc}{{\np}-{complete}}

\newcommand{\nphshort}{{\np}-{h}}
\newcommand{\bnphshort}{{\textbf{{\np}-{h}}}}

\newcommand{\wa}{\textsf{\mbox{W[1]}}}
\newcommand{\wb}{\textsf{W[2]}}

\newcommand{\wbh}{{\textsf{W[2]}}-hard}
\newcommand{\wbhshort}{{{\textsf{W[2]}}-h}}
\newcommand{\bwbhshort}{{\textbf{\textsf{W[2]}-h}}}

\newcommand{\wbc}{\wb-complete}
\newcommand{\wbhns}{{\textsf{W[2]}}-hardness}
\newcommand{\fpt}{\textsf{\mbox{FPT}}}

\newcommand{\prob}[1]{{\textsc{#1}}}
\newcommand{\setmid}{:}
\newcommand{\abs}[1]{|#1|}
\newcommand{\edge}[2]{\{#1, #2\}}
\newcommand{\Succ}{\,\ }
\newcommand{\wahns}{\textsf{W[1]}-hardness}
\newcommand{\wahshort}{\textsf{W[1]}-h}

\newcommand{\onlyconf}[1]{}
\newcommand{\remove}[1]{}
\newcommand{\memph}[1]{\emph{#1}}

\newcommand{\inneighbor}[2]{\Gamma_{#1}^-(#2)}

\newcommand{\outneighbor}[2]{\Gamma_{#1}^+(#2)}

\newcommand{\neighbor}[2]{\Gamma_{#1}(#2)}

\newcommand{\smallo}[1]{o(#1)}
\newcommand{\bs}{B}
\newcommand{\rs}{R}

\newcommand{\arc}[2]{(#1, #2)}

\newcommand{\ac}{\text{AC}}
\newcommand{\av}{\text{AV}}
\newcommand{\dv}{\text{DV}}
\newcommand{\dc}{\text{DC}}

\newcommand{\dtime}{\textsf{DTIME}}

\newcommand{\EP}[3]{
\begin{center}
\renewcommand{\tabcolsep}{0.5mm}
{
\begin{tabular}{|lp{0.88\columnwidth}|}\hline
\multicolumn{2}{|l|}{\textsc{#1}} \\ \hline
{\bf Given:}    & #2  \\
{\bf Question:} & #3  \\ \hline
\end{tabular}
}
\end{center}
}

\usepackage{geometry}
\begin{document}
\title{On the Complexity of the Two-Stage Majoritarian Rule\thanks{A preliminary version of the paper appears in the Proceedings of the 22nd International Conference on Autonomous Agents and Multiagent Systems (AAMAS~2023)~\protect\cite{DBLP:conf/atal/000123a}.}}
\author{Yongjie Yang}
\date{\small{Saarland University\\
yyongjiecs@gmail.com}
}

\maketitle

\begin{abstract}
Sequential voting rules have played a crucial role in shaping decisions within parliamentary and legislative frameworks. 
After observing that the existing sequential rules fail several fundamental axioms, Horan and Sprumont [2022] proposed a sequential rule named two-stage majoritarian rule (TSMR). This paper examines this rule by investigating the complexity of {\sc{Agenda Control}}, {\sc{Coalition Manipulation}}, {\sc{Possible Winner}}, {\sc{Necessary Winner}}, and eight standard election control problems. Our study offers a comprehensive insight into the complexity landscape of these problems.
%
\end{abstract}



\section{Introduction}
Exploring the complexity of strategic voting problems has been a vibrant topic in computational social choice (see, e.g.,~\cite{DBLP:conf/ijcai/BoehmerBFN21,DBLP:journals/ai/ElkindGORV21,DBLP:conf/ijcai/FitzsimmonsH22,DBLP:journals/jcss/HemaspaandraHR22,DBLP:journals/ai/NevelingR21}). This pursuit stems from the recognition that malicious strategic voting can potentially undermine the integrity of election results, and it is widely acknowledged that complexity could serve as a barrier against strategic actions~\cite{Bartholdi92howhard,BARTHOLDI89}. In particular, to what extent a voting rule resists strategic voting is commonly recognized as a pivotal criterion for evaluating its applicability. 
Over the past three decades, the complexity of many different strategic voting problems under numerous voting rules has been established~\cite{BaumeisterR2016,handbookofcomsoc2016Cha7FR}. 
Against this backdrop, whenever a new voting rule with desirable axiomatic properties is proposed, comparing it with existing rules in terms of resistance to strategic voting is of great importance. 

This paper aims to complete the complexity landscape of several strategic voting problems under a sequential voting rule recently proposed by Horan and Sprumont~\cite{HoranSprumont2022}. 
Given the preferences of a set of voters on a candidate set and an agenda (a linear ordering specifying the order in which candidates are considered), a sequential rule outputs a single winner. 
Sequential rules are exceedingly useful in parliamentary and legislative decision-making. 
The amendment rule and the successive rule have long been among the most widely used sequential rules in many countries~\cite{Rasch2000}. 
Under the amendment rule, the procedure unfolds over multiple rounds, each producing a temporary winner. The first candidate in the agenda wins the first round. In round~$i$, the $i$-th candidate is compared head-to-head with the winner of round $i-1$, and the candidate preferred by a majority becomes the new temporary winner. The winner of the final round is the amendment winner. 
In contrast, under the successive procedure, the winner is the first candidate in the agenda who is preferred by a majority of voters over all subsequent candidates.

Despite their practical relevance, both the amendment rule and the successive rule suffer from several theoretical limitations. Recently, Horan and Sprumont~\cite{HoranSprumont2022} proposed a new sequential voting rule called the \emph{two-stage majoritarian rule} (TSMR), and provided a comprehensive axiomatic characterization. 
Given the voters' preferences, the majority graph is a directed graph (digraph) whose vertices correspond to candidates, with an arc from one candidate to another whenever a majority of voters prefer the former to the latter. 
The TSMR operates in two stages. In the first stage, it identifies the sources (i.e., candidates with no incoming arcs) of the acyclic spanning subgraph formed by the forward arcs of the majority graph with respect to the agenda. These are precisely the candidates that do not lose a head-to-head comparison to any of their predecessors. This stage can be viewed as a shortlisting phase. 
In the second stage, the rule selects the last candidate among all shortlisted sources, according to the agenda, as the winner. Therefore, the winner does not lose any head-to-head comparison within the shortlist. Notably, when the number of voters is odd---a common and practically relevant assumption---the winner is preferred by a majority of voters in every pairwise comparison within the shortlist.

Although TSMR satisfies several desirable axiomatic properties, it remains unclear how it performs under strategic behavior, which is a critical consideration in the design and evaluation of voting rules. Since the amendment rule and the successive rule are widely used and have been extensively studied with regard to their vulnerability or resistance to manipulation, the introduction of TSMR by Horan and Sprumont~\cite{HoranSprumont2022} naturally raises the question of how the new rule compares to these benchmarks in terms of resistance to strategic actions.  
This paper endeavors to provide insights into this inquiry. In addition, to achieve a more comprehensive understanding of TSMR, we study two winner determination problems under scenarios where only partial information about voters' preferences is available.

Our main contributions are as follows.

\begin{enumerate}
    \item[(1)] We study the {\prob{Agenda Control}} problem, where an external agent, empowered to set the agenda, attempts to ensure that a distinguished candidate wins. More precisely, given an election and a preferred candidate, the agent must determine how to arrange the agenda so that this candidate emerges as the winner.

\item[(2)] We study the {\prob{Coalition Manipulation}} problem. Unlike {\prob{Agenda Control}}, this problem assumes a fixed agenda. A subset of voters, called manipulators, seeks to make a distinguished candidate win by strategically coordinating their votes. That is, the task is to determine how the manipulators should (mis)report their preferences to secure the victory of the preferred candidate.

    \item[(3)] We study eight standard election control problems, namely, {\prob{CCAV}}, {\prob{CCDV}}, {\prob{CCAC}}, {\prob{CCDC}}, {\prob{DCAV}}, {\prob{DCDV}}, {\prob{DCAC}}, and {\prob{DCDC}}. In these abbreviations, ``CC''/``DC'' stands for ``constructive control''/``destructive control'', the third letter ``A''/``D'' stands for ``adding''/``deleting'', and the last letter ``V''/``C'' stands for ``votes''/``candidates''. These problems model a scenario where an external agent aims to make a distinguished candidate win (constructive) or not win (destructive) by adding or deleting a limited number of votes or candidates.

    \item[(4)] We study the {\prob{Possible Winner}} and the {\prob{Necessary Winner}} problems, which arise when only partial information on voter preferences or agenda is available. {\prob{Possible Winner}} asks which candidates can win at least one completion of the partial input, whereas {\prob{Necessary Winner}} asks which candidates win in every completion regardless of the missing information.

    \item[(5)] For the above problems, we provide an exhaustive examination of their (parameterized) complexity. 
    For the eight election control problems, as well as the {\prob{Possible Winner}} and the {\prob{Necessary Winner}} problems, we study both the special case where the distinguished candidate~$p$ is the first in the agenda and the case where~$p$ is the last. We also discuss two parameterizations that capture the degree of incompleteness in the input for the two winner determination problems under partial information. 
    For a concise overview of our findings and comparisons with previous results for the amendment rule and the successive rule, see Table~\ref{tab-results}. 

    \item[(6)] As byproducts, our reductions also provide numerous lower bounds for kernelization algorithms, approximation algorithms, and exact algorithms for the aforementioned problems.
\end{enumerate}

 \begin{table}[ht!]
 \caption{The complexity of various voting problems under different sequential rules, with our findings highlighted in bold. The terms ``first'', ``last'', and ``$\overline{\text{last}}$'' denote the position of the distinguished candidate in the agenda, representing first, last, and not last, respectively. {\poly}-results hold regardless of the distinguished candidate's position. Additionally,~$m$ is the number of candidates,~$n$ is the number of votes,~$n_{\text{rg}}$ is the number of registered votes, and~$k$ is the number of votes/candidates allowed to be added or deleted. The hardness results for the {\prob{Possible Winner}} and {\prob{Necessary Winner}} problems hold under further restrictions, such as when the distinguished candidate appears in specific positions and each partial vote contains only a constant number of undetermined pairs. See Section~\ref{sec-possible-necessary} for further details.}
    \label{tab-results}
    \begin{center}
    \renewcommand{\tabcolsep}{0.8mm}
    \footnotesize{
    \begin{tabular}{llllll}\toprule
    \multicolumn{2}{l}{}
    & CCAV
    & CCDV
    & CCAC
    & CCDC\\ \midrule

    {TSMR}
    & first
    & {\bwbhshort} ($k+{n_{\text{rg}}}$, Thm.~\ref{thm-ccav-hard})
    & {\bwbhshort} ($k$, $n-k$, Thms.~\ref{thm-ccdv-hard-deleted-first},~\ref{thm-ccdv-hard-first})
    & {{\bwbhshort} ($k$, Thm.~\ref{thm-ccac-hard})}
    & {{\bpoly} (Thm.~\ref{thm-ccdc-p})}  \\ \cline{2-6}

    & last
    & {\bwbhshort} ($k+{n_{\text{rg}}}$, Thm.~\ref{thm-ccav-hard-last})
    & {\bwbhshort} ($k$, $n-k$, Thms.~\ref{thm-ccdv-hard},~\ref{thm-ccdv-wbh-not-deleted-last})
    & {immune (Cor.~\ref{cor-ccac-last-immne})}
    & {{\bpoly} (Thm.~\ref{thm-ccdc-p})}\\ \midrule

  {amendment {\cite{DBLP:journals/tcs/LiuFZL09,DBLP:journals/algorithmica/Yang25}}}
    &first
    & {\wahshort} ($k+{n_{\text{rg}}}$)
    & {\wahshort} ($k$)
    & {\poly}
    & {\poly}
    \\ \cline{2-6}

    & last
    & {\wbhshort} ($k+{n_{\text{rg}}}$)
    & {\wbhshort} ($k$)
    & {\poly}
    & {\poly}\\ \midrule

    {successive {\cite{DBLP:journals/tcs/LiuFZL09,DBLP:journals/algorithmica/Yang25}}}
    & first
    & {\poly}
    & {\poly}
    & {immune}
    & {\wahshort} ($k$, $m-k$) \\ \cline{2-6}

    &last
    & {\wahshort} ($k+n_{\text{rg}}$)
    & {\wbhshort} ($k$)
    & {{\wbhshort} ($k$)}
    & {\poly} 
    \\  \bottomrule
    \end{tabular}
    }
\medskip

    \renewcommand{\tabcolsep}{1.08mm}
    \footnotesize{
    \begin{tabular}{llllll}\toprule
     \multicolumn{2}{l}{}
     &DCAV
     &DCDV
     &DCAC
     &DCDC\\ \midrule

    {\multirow{2}{*}{TSMR}}
    & $\overline{\text{last}}$
    & {\bwbhshort} ($k+{n_{\text{rg}}}$, Thm.~\ref{thm-dcav-hard-first})
    & {\bwbhshort} ($k$, $n-k$, Thms.~\ref{thm-dcdv-hard-deleted-first},~\ref{thm-dcdv-wbh-not-deleted-first})
    & {{\bpoly} (Thm.~\ref{thm-dcac-p})}
    & {{\bpoly} (Cor.~\ref{cor-dcdc-p})}
    \\ \cline{2-6}

    & last
    & {\poly} \cite{Bartholdi92howhard}
    & {\poly} \cite{Bartholdi92howhard}
    & {{\bpoly} (Thm.~\ref{thm-dcac-p})}
    & immune (trivial) \\ \midrule

{amendment} {\cite{DBLP:journals/tcs/LiuFZL09,DBLP:journals/algorithmica/Yang25}}
    & first
    & {\poly}
    & {\poly}
    & {\poly}
    & immune
    \\ \cline{2-6}

    & last
    & {\wahshort} ($k$)
    & {\wbhshort} ($k$)
    & {\poly}
    & {\poly} \\ \midrule

    {successive} {\cite{DBLP:journals/tcs/LiuFZL09,DBLP:journals/algorithmica/Yang25}}
    & first
    & {\poly}
    & {\poly}
    & {\wbhshort} ($k$)
    & immune
    \\ \cline{2-6}

    & last
    & {\poly}
    & {\poly}
    & {\poly}
    & {\wahshort} ($k$, $m-k$) \\  \bottomrule
    \end{tabular}
}
\medskip

    \renewcommand{\tabcolsep}{1.45mm}
    \footnotesize{
    \begin{tabular}{lllll}\toprule
     & {\prob{Agenda Control}}
     & {\prob{Coalition Manipulation}}
     & {\prob{Possible Winner}}
     & {\prob{Necessary Winner}}\\ \midrule

     {TSMR}
     & {\bpoly} (Thm.~\ref{thm-agenda-control-p-general})
     & {\bpoly} (Thm.~\ref{thm-manipulation-p})
     & {\bnphshort} (Thms.~\ref{thm-possible-hard},~\ref{thm-possible-hard-second-last})
     & {\bpoly} (Thm.~\ref{thm-necessary-p}) \\ \midrule

  amendment
     & {\poly}
     & {\poly}
     & {\nphshort}
     & {\conphshort} \\ \midrule

     successive {\cite{DBLP:journals/jair/BredereckCNW17}}
     &{\poly}
     & {\poly}
     & {\nphshort}
     & {\poly} \\  \bottomrule
    \end{tabular}
    }
        \end{center}
\end{table}

\subsection{Related Works}
{\prob{Agenda Control}} is one of the most extensively studied problems in the context of sequential voting (see, e.g.,~\cite{Black1958,Miller1977}). Despite its significance, the complexity of {\prob{Agenda Control}} remained uncharted territory until recently~\cite{DBLP:journals/jair/BredereckCNW17}. We note that the complexity of analogous problems within the domain of knockout tournaments was studied earlier~\cite{DBLP:journals/ai/AzizGMMSW18,BARTHOLDI89,Bartholdi92howhard,DBLP:conf/ijcai/ConitzerLX09,DBLP:journals/dam/ManurangsiS23,DBLP:conf/aaai/Williams10,williamstournamentinhandbook2016}. 

{\prob{Coalition Manipulation}}, initially scrutinized by Conitzer, Sandholm, and Lang~\cite{DBLP:journals/jacm/ConitzerSL07}, emerges as a natural generalization of the renowned {\prob{Manipulation}} problem~\cite{BARTHOLDI89}. We refer to~\cite{BaumeisterR2016,Conitzer2016,DBLP:journals/aarc/Veselova16,DBLP:journals/amai/Walsh11,DBLP:journals/jair/Walsh11} for detailed results on the complexity of this problem for many traditional rules (e.g., Borda, Maximin, etc., which do not necessitate an agenda for determining the winner).

The constructive election control problems were first studied by Bartholdi, Tovey, and Trick~\cite{Bartholdi92howhard}, while their destructive counterparts were introduced by Hemaspaandra, Hemaspaandra, and Rothe~\cite{DBLP:journals/ai/HemaspaandraHR07}. Since then, the computational complexity of these problems has been extensively investigated for a wide range of voting rules, including the amendment rule and the successive rule~\cite{DBLP:journals/algorithmica/Yang25}. For a comprehensive overview of progress up to 2016, see~\cite{BaumeisterR2016,handbookofcomsoc2016Cha7FR}, and for recent results, see~\cite{DBLP:journals/jair/CarletonCHNTW24,DBLP:conf/atal/CarletonCHNTW23a,DBLP:journals/aamas/ErdelyiNRRYZ21,DBLP:journals/ai/NevelingR21,DBLP:conf/atal/Yang17,AAMAS17YangBordaSinlgePeaked,DBLP:conf/ecai/000120}.

The complexity of {\prob{Possible Winner}} and {\prob{Necessary Winner}} for the amendment rule and the successive rule has been established by Bredereck~et~al.~\cite{DBLP:journals/jair/BredereckCNW17}. For traditional voting rules, these problems were first studied by Konczak and Lang~\cite{Konczak05votingprocedures}, and their complexity for many additional rules has been subsequently analyzed~\cite{DBLP:journals/tdasci/ChakrabortyDKKR21,DBLP:conf/atal/ChakrabortyK21,DBLP:journals/jair/XiaC11}.

\subsection{Organization}
The remainder of the paper is organized as follows. Section~\ref{sec-pre} presents the formal definitions of key concepts utilized throughout the paper. In Section~\ref{sec-strategic}, we unfold our concrete findings regarding strategic problems, including {\prob{Agenda Control}}, {\prob{Coalition Manipulation}}, and the eight standard election control problems. Following this, Section~\ref{sec-possible-necessary} delves into the {\prob{Possible Winner}} and the  {\prob{Necessary Winner}} problems. In Section~\ref{sec-lowerbounds}, we elaborate on how our reductions provide significant insights into other algorithmic lower bounds. Finally, Section~\ref{sec-con} offers a summary of our results and outlines potential avenues for future research.

Proofs of certain results are deferred to the appendix to preserve the flow of the main narrative. Such statements are marked with a~$\bigstar$.

\section{Preliminaries}
\label{sec-pre}
We assume the reader is familiar with the basics of graph theory, complexity theory, and parameterized complexity theory~\cite{DBLP:books/sp/BG2018,DBLP:books/sp/CyganFKLMPPS15,DBLP:conf/lata/Downey12,DBLP:journals/interfaces/Tovey02}. 

For an integer~$i$, let~$[i]$ denote the set of all positive integers no greater than~$i$. 
For a binary relation~$\rs$, we often use~$x\, \rs\, y$ to denote $(x, y)\in \rs$.

\subsection{Graphs}

An undirected graph is a tuple $G=(N, A)$, where~$N$ is the set of vertices and~$A$ is the set of edges. An edge between vertices $v$ and $v'$ is denoted by $\edge{v}{v'}$. The set of neighbors of~$v$ in~$G$ is $\neighbor{G}{v}=\{v'\in N \setmid \edge{v}{v'}\in A\}$. A vertex~$v$ dominates a vertex~$v'$ if $v=v'$ or $\edge{v}{v'}\in A$.  For $S, S' \subseteq N$, $S$ dominates $S'$ if every vertex in $S'$ is dominated by at least one vertex in $S$. 

A digraph is a tuple $G=(N, A)$, where~$N$ denotes the set of vertices, and~$A$ represents the set of arcs. 
An arc from a vertex~$a$ to a vertex~$b$ is denoted as~$\arc{a}{b}$. The set of inneighbors of a vertex~$a$ in~$G$ is~$\inneighbor{G}{a}=\{b\in N \setmid \arc{b}{a}\in A\}$, while the set of outneighbors of~$a$ in~$G$ is~$\outneighbor{G}{a}=\{b\in N \setmid \arc{a}{b}\in A\}$. For~$S\subseteq N$, let $\outneighbor{G}{S}=\bigcup_{a\in S}\outneighbor{G}{a} \setminus S$. In contexts where the specific graph~$G$ is evident, we may drop~$G$ from the notations. 
An oriented graph is a digraph in which there is at most one arc between any pair of vertices. 

For a graph $G$ (directed or undirected) and a subset $S \subseteq N$, the subgraph of $G$ induced by~$S$, denoted $G[S]$, is the graph with vertex set $S$ consisting of all edges (or arcs) in $G$ that have both endpoints in $S$.

\subsection{Elections}
An election is a tuple $(C, V)$ of a set of candidates~$C$ and a multiset of votes~$V$. Each vote $\succ\, \in V$ is defined as a linear order on~$C$. For two candidates $c, c'\in C$, we denote that~$c$ is ranked before~$c'$ in a vote~$\succ$ if $c\succ c'$. Furthermore,~$c$ is ranked immediately before~$c'$ in~$\succ$ if $c\succ c'$ and no candidate is ranked between them. A vote $\succ$ expresses the voter’s preference:~$a$ is preferred to~$b$ if $a \succ b$.  
Preferences can also be written as sequences; for example, $a \Succ b \Succ c$ indicates that~$a$ is ranked before~$b$, and~$b$ before~$c$. 

For an election $(C, V)$ and a subset $C' \subseteq C$ of candidates, let $(C', V_{C'})$ denote the election restricted to~$C'$, obtained from $(C, V)$ by removing candidates in $C \setminus C'$ from both the candidate set and all votes in~$V$. When the context is clear, we simply write $(C', V)$ instead of $(C', V_{C'})$.

An agenda~$\rhd$ on~$C$ is a linear order on~$C$. 
If $c \rhd c'$, we say $c$ precedes $c'$ (or $c$ is a predecessor of $c'$), and $c'$ succeeds $c$ (or $c'$ is a successor of $c$).

A sequential rule~$\tau$ maps an election $(C, V)$ and an agenda~$\rhd$ to a single winner $\tau(C, V, \rhd)\in C$.

For $c, c'\in C$, we use~$n_V(c, c')$ to denote the number of votes in~$V$ ranking~$c$ before~$c'$. We assert that~$c$ beats (respectively, ties with)~$c'$ with respect to~$V$ if $n_V(c, c')>n_V(c', c)$ (respectively,  $n_V(c, c')=n_V(c', c)$). A candidate~$c\in C$ is a weak Condorcet winner of~$(C, V)$  if no other candidate beats~$c$ with respect to $V$, and a Condorcet winner if $c$ beats all other candidates.

\subsection{Majority Graphs}
The majority graph of an election $\elec = (C, V)$, denoted by~$G_E$, is an oriented graph with vertex set~$C$. There is an arc from $c \in C$ to $c' \in C$ in~$G_E$ if and only if $n_V(c, c') > n_V(c', c)$. This structure is not only central to the definition of the TSMR but also serves as the starting point for several of our polynomial-time algorithms. To facilitate our discussion, we first provide the time complexity for computing the majority graph.

Constructing the majority graph for an election with $m$ candidates and $n$ votes can be done straightforwardly in time $\bigo{n\cdot m^2}$. 
Sornat, Vassilevska Williams, and Xu~\cite{DBLP:conf/sigecom/SornatWX21} identified a strong connection between the computation of the majority graph and the well-studied {\prob{Dominance Product}} problem, initially introduced by Matou\v{s}ek~\cite{DBLP:journals/ipl/Matousek91}. Exploiting this connection, they developed a faster algorithm, summarized in the following lemma.\footnote{Notably, the algorithm indeed applies to the computation of weighted majority graphs.}

\begin{lemma}[\cite{DBLP:conf/sigecom/SornatWX21,DBLP:journals/ipl/Matousek91}]
\label{lemma-majority-graph}
Given an election with $m$ candidates and $n$ votes, its majority graph can be computed in time~$\bigo{n \cdot m^{1.69} + n^{0.55} \cdot m^2}$.
\end{lemma}

The running time bound in Lemma~\ref{lemma-majority-graph} depends on the current upper bound on the exponent of square matrix multiplication, known to be at most~$2.373$. Any future improvement in matrix multiplication algorithms would immediately lead to an updated bound in Lemma~\ref{lemma-majority-graph}. We refer the reader to~\cite{DBLP:journals/theoretics/AlmanW24} for further details.

\subsection{Sequential Voting Rules}
We now formally define the TSMR, the amendment rule, and the successive rule. Although the latter two are not the primary focus of this paper, they are closely intertwined with our discussions, and we therefore provide their formal definitions as well. 

Let $E=(C, V)$ be an election and~$\rhd$ an agenda on~$C$.

\begin{itemize}
\item {\bf{Two-stage majoritarian rule (TSMR)}}\footnote{Horan and Sprumont \cite{HoranSprumont2022} defined a more general TSMR that accepts as input the majority graph of an election $(C, V)$ along with a subset $C'\subseteq C$, and returns a winner from~$C'$. Our formulation aligns with the scenario where $C'=C$.} 
Let~$G_E^{\rhd}$ be the spanning subdigraph of the majority graph~$G_E$ of~$E$ with only forward arcs with respect to~$\rhd$. In other words,~$G_E^{\rhd}$ has the vertex set~$C$, and there exists an arc from~$c$ to~$c'$ in~$G_E^{\rhd}$ if and only if $c\rhd c'$ and there is an arc from~$c$ to~$c'$ in~$G_E$. Let $C'\subseteq C$ be the set of sources of $G_E^{\rhd}$, i.e., candidates without inneighbors. Then, the winner is the rightmost candidate in~$C'$, i.e., the one $c\in C'$ such that $c'\rhd c$ for all $c'\in C'\setminus \{c\}$.

\item {\bf{Amendment}} This procedure unfolds across~$\abs{C}$ rounds, each round determining a temporary winner. Precisely, the winner of the first round is the first candidate in the agenda. The winner of round~$i$ where $i\geq 2$ is determined as follows. Let~$c$ be the winner of round~$i-1$, and let~$c'$ be the $i$-th candidate in the agenda. The winner of round~$i$ is~$c$ if~$c$ beats~$c'$; otherwise, it is~$c'$. The amendment winner is the winner of the last round.

\item {\bf{Successive}}
For a candidate $c\in C$ and a subset $C'\subseteq C\setminus \{c\}$, we say that~$c$ majority-dominates~$C'$ in~$E$ if a majority of votes in~$V$ rank~$c$ before all candidates in~$C'$. The successive winner is the first candidate in the agenda who majority-dominates the set of all her successors.
\end{itemize}

The amendment rule and the successive rule have also been studied under other names (cf.~\cite{Black1958,Farquharson1969}). 
Note that all three sequential rules are resolute (decisive)---they always return a single winner. This resoluteness arises from the fact that the agenda is part of the input, introducing an asymmetry in the ordering of candidates. While this guarantees a unique outcome, it comes at the cost of sacrificing neutrality.

\begin{example} Let $C=\{a, b, c, d\}$, and let~$V$ be a set of three votes with the preferences \underline{$b\Succ d\Succ c\Succ a$}, \underline{$c\Succ a\Succ b\Succ d$}, and \underline{$a\Succ d\Succ b\Succ c$}, respectively. 
The majority graph of~$(C, V)$, three different agendas, and the winners under different rules and agendas are shown in Figure~\ref{fig-seqrule}. 
\end{example}
\begin{figure}[ht!]
\centering{
\includegraphics[width=0.98\textwidth]{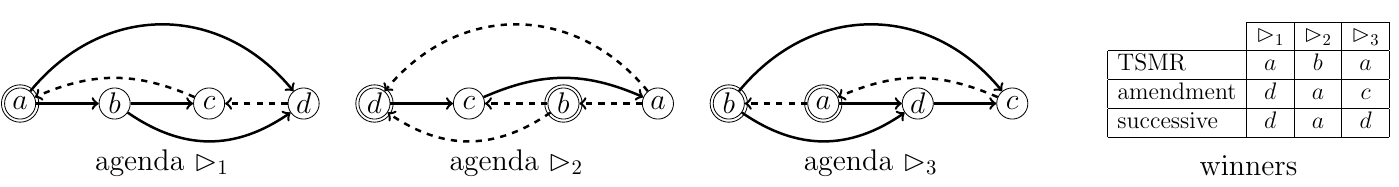}
}
\caption{An illustration of TSMR, the amendment rule, and the successive rule. Dashed arcs represent backward arcs with respect to~$\rhd_i$ (i.e., arcs not included in~$G_E^{\rhd_i}$). Double-circled vertices represent the sources of the spanning subdigraph of the majority graph that contains only forward arcs.}
\label{fig-seqrule}
\end{figure}

The definitions of the sequential rules readily reveal a connection between the first and last candidates in an agenda and the (weak) Condorcet winner, as outlined below.

\begin{observation}
\label{obs-a}
For an election~$(C, V)$ and an agenda~$\rhd$ on~$C$, the following hold.
\begin{enumerate}
    \item[\emph{(1)}] The first candidate in~$\rhd$ is the amendment winner of~$(C, V)$ if and only if it is the Condorcet winner of~$(C, V)$.
    \item[\emph{(2)}] The last candidate in~$\rhd$ is the TSMR winner of~$(C, V)$ if and only if it is a weak Condorcet winner of~$(C, V)$.
    \item[\emph{(3)}] If the first candidate in~$\rhd$ is the successive winner of~$(C, V)$, then it is also the Condorcet winner of~$(C, V)$.
    \item[\emph{(4)}] If the first candidate in~$\rhd$ is the Condorcet winner of~$(C, V)$, then it is also the TSMR winner of~$(C, V)$.
    \item[\emph{(5)}] If the last candidate in~$\rhd$ is a weak Condorcet winner of~$(C, V)$, then it is also the amendment winner and the successive winner of~$(C, V)$.
    \item[\emph{(6)}] The converses of (3)--(5) do not necessarily hold.
\end{enumerate}
\end{observation}

\subsection{Problem Formulations}
For a sequential voting rule~$\tau$, we study the following problems defined in~\cite{DBLP:journals/jair/BredereckCNW17}.

\EP{Agenda Control-$\tau$}
{An election~$(C, V)$ and a distinguished candidate $p\in C$.}
{Is there an agenda~$\rhd$ on~$C$ such that~$p$ is the winner of $(C, V, \rhd)$ with respect to~$\tau$, i.e., $p=\tau(C, V, \rhd)$?}

\EP{Coalition Manipulation-$\tau$}
{An election $(C, V)$, a distinguished candidate $p\in C$, an agenda~$\rhd$ on~$C$, and a positive integer~$k$.}
{Is there a multiset~$V'$ of~$k$ votes on~$C$ such that $p=\tau(C, V\cup V', \rhd)$?}

For a partial order~${\rs}$ on a set~$X$, a linear extension of ${\rs}$ is a linear order on~$X$ containing~${\rs}$, i.e., a linear order~${\rs}'$ such that $(x,y)\in {\rs}$ implies $(x, y)\in {\rs}'$ for all $x, y\in X$.

A partial election is a tuple $(C, V)$, where~$V$ is a multiset of partial orders on~$C$. A completion of $(C, V)$ is an election $(C, V')$ where elements of~$V'$ one-to-one correspond to elements of~$V$ such that every $v'\in V'$ is a linear extension of its corresponding partial order in~$V$. A partial agenda on~$C$ is a partial order on~$C$.

\EP{Possible Winner-$\tau$}
{A partial election $(C, V)$, a distinguished candidate $p\in C$, and a partial agenda~$\rhd$ on~$C$.}
{Is there a completion $(C, V')$ of $(C, V)$ and a linear extension~$\rhd'$ of~$\rhd$ such that $p=\tau(C, V', \rhd')$?}

\EP{Necessary Winner-$\tau$}
{A partial election $(C, V)$, a distinguished candidate $p\in C$, and a partial agenda~$\rhd$ on~$C$.}
{Does $p=\tau(C, V', \rhd')$ hold for all completions~$(C, V')$ of~$(C, V)$ and all linear extensions~$\rhd'$ of~$\rhd$?}

We also study eight standard control problems which are special cases of the following problems.

\EP{Constructive Multimode Control-$\tau$}
{An election $(C \cup D,V \cup W)$ with a set~$C$ of registered
  candidates,
  a set~$D$ of unregistered candidates, a multiset~$V$ of
  registered votes, a multiset~$W$ of unregistered votes, a
  distinguished candidate $p \in C$, an agenda~$\rhd$ on $C\cup D$, and four
  integers~$k_{\av}$,~$k_{\dv}$,~$k_{\ac}$, and~$k_{\dc}$ with~$k_{\av}\leq |W|$, $k_{\dv}\leq |V|$, $k_{\ac} \leq |D|$, and $k_{\dc} \leq |C|$.
  }
{Are there $V' \subseteq V$, $W' \subseteq W$, $C' \subseteq
  C\setminus \{p\}$, and $D' \subseteq D$ such that $\abs{V'} \leq k_{\dv}$, $\abs{W'}\leq k_{\av}$, $\abs{C'}\leq k_{\dc}$,
  $\abs{D'}\leq k_{\ac}$, and~$p$ wins
  $((C \setminus C')\cup D', (V \setminus V') \cup W', \rhd')$ with respect to~$\tau$, where $\rhd'$ is~$\rhd$ restricted to $(C \setminus C')\cup D'$?}

{\sc{Destructive Multimode Control-$\tau$}} has the same input as {\sc{Constructive Multimode Control-$\tau$}}, and asks whether there
exist~$V'$,~$W'$,~$C'$, and~$D'$ (as defined above) such that~$p$ is not the~$\tau$ winner of $((C \setminus C')\cup D',(V \setminus V') \cup W', \rhd')$.

\begin{table}[ht!]
\caption{Special cases of \prob{Constructive/Destructive Multimode Control}. Here,~X represents either CC for constructive control or DC for destructive control.}
\begin{center}
\begin{tabular}{l  l} \toprule
problems & restrictions \\ \midrule

XAV-$\tau$ & $k_{\dv}=k_{\ac}=k_{\dc}=0$, $D=\emptyset$ \\

XAC-$\tau$ & $k_{\av}=k_{\dv}=k_{\dc}=0$, $W=\emptyset$\\

XDV-$\tau$ & $k_{\av}=k_{\ac}=k_{\dc}=0$, $D=W=\emptyset$\\ 

XDC-$\tau$ & $k_{\av}=k_{\dv}=k_{\ac}=0$, $D=W=\emptyset$\\ \bottomrule
\end{tabular}
\end{center}
\label{tab-def-restriction}
\end{table}

Table~\ref{tab-def-restriction} summarizes the eight standard control problems as special cases of the above two problems. For simplicity, in our examination of these problems, we use~$k$ to denote the integer component of the input that is not mandated to be $0$, omitting those that are requested to be~$0$ or~$\emptyset$. As an illustration, an instance of {\prob{CCAV-$\tau$}} is represented as $((C, V\cup W), p, \rhd, k)$, where~$k$ denotes~$k_{\text{AV}}$. We may occasionally refrain from explicitly stating that candidates (or votes) from~$C$ (or~$V$) are registered, and those from~$D$ (or~$W$) are unregistered. For readability, the suffix ``$\tau$'' may be dropped when no ambiguity arises.

Our hardness results are established through reductions from the following problem.

\EP{Red-Blue Dominating Set (\prob{RBDS})}
{A bipartite graph~$G$ with bipartition $(\rs, \bs)$, where the vertices in~$\rs$ and~$\bs$ are called red and blue vertices, respectively, and a positive integer~$\kappa \leq \abs{\bs}$.}
{Is there a subset $\bs' \subseteq \bs$ of cardinality~$\kappa$ that dominates~$\rs$ in~$G$?}

{\prob{RBDS}} is {\npc}~\cite{garey}, and is {\wbc} with respect to~$\kappa$~\cite{fellows2001}.

\subsection{Other Useful Notations}

Unless stated otherwise, for a set~$S$,~$\overrightarrow{S}$ denotes an arbitrary fixed linear ordering on~$S$, and~$\overleftarrow{S}$ denotes its reverse. For $S'\subseteq S$, $\overrightarrow{S}[S']$ denotes the restriction of~$\overrightarrow{S}$ to~$S'$, and $\overrightarrow{S}\setminus S'$ denotes~$\overrightarrow{S}$ with elements of~$S'$ removed.

\subsection{Remarks}
Previous studies on sequential voting rules often assume that the number of voters is odd, allowing majority graphs to be treated as tournaments and simplifying the analysis (see, e.g.,~\cite{DBLP:journals/jair/BredereckCNW17,HoranSprumont2022}). This assumption is also practically relevant, as many real-world parliamentary decision-making processes are intentionally designed to involve an odd number of voters to avoid ties. 
To ensure generality, our results are presented without relying on this assumption, and they remain valid when the condition is imposed. In particular, our polynomial-time algorithms apply directly. For the hardness results, several of our reductions can be adapted with minor modifications, while others already apply to this setting. We note that for the voter control problems, this assumption implies that no ties arise after adding or deleting votes. Our reductions for candidate control problems are robust in this regard, as they ensure that the assumption holds in the constructed elections.

All our reductions are polynomial-time, and all computationally hard problems proved in the paper are in {\np}. Consequently, any problem shown to be {\wbh} in the paper is also {\npc}. For brevity, NP-completeness is not stated explicitly in the corresponding theorems. 

\section{Strategic Problems}
\label{sec-strategic}
In this section, we study the complexity of several strategic voting problems for TSMR. We begin with a result on the winner determination problem, since this procedure serves as a key subroutine in several results discussed later. 

\begin{lemma}
\label{lem-TSMR-winner-determination}
Given an election $E = (C, V)$ with $m$ candidates and $n$ votes, and an agenda~$\rhd$ on~$C$, the TSMR winner of~$E$ with respect to~$\rhd$ can be computed in time~$\bigo{n \cdot m^{1.69} + n^{0.55} \cdot m^2}$.
\end{lemma}

\begin{proof}
By Lemma~\ref{lemma-majority-graph}, the majority graph of~$E$ can be computed in time~$\bigo{n \cdot m^{1.69} + n^{0.55} \cdot m^2}$. 
To construct the spanning subdigraph~$G_E^{\rhd}$---consisting of the forward arcs of the majority graph according to the agenda~$\rhd$---we examine at most~${m\cdot (m-1)}/{2}$ arcs, yielding a construction time of~$\bigo{m^2}$. 
Finally, identifying the last source of~$G_E^{\rhd}$ with respect to~$\rhd$ takes~$\bigo{m}$ time. 
Thus, the overall complexity remains 
$\bigo{n \cdot m^{1.69} + n^{0.55} \cdot m^2}$.
\end{proof}

\subsection{Agenda Control and Coalition Manipulation}
We now investigate the complexity of the {\prob{Agenda Control}} and {\prob{Coalition Manipulation}} problems for TSMR. We show that both problems can be solved within the same time complexity. We first consider {\prob{Agenda Control}}.

\begin{theorem}
\label{thm-agenda-control-p-general}
{\prob{Agenda Control-TSMR}} can be solved in time $\bigo{n \cdot m^{1.69} + n^{0.55} \cdot m^2}$, 
where~$m$ is the number of candidates and~$n$ is the number of votes in the election.
\end{theorem}

\begin{proof}
Let $((C, V), p)$ be an instance of {\prob{Agenda Control-TSMR}} with $m = \abs{C}$ and $n = \abs{V}$. Let~$G$ be the majority graph of $(C, V)$. Let $A=C\setminus (\inneighbor{G}{p}\cup \{p\})$ be the set of candidates from $C\setminus \{p\}$ not beating~$p$ with respect to~$V$. 
We construct an agenda~$\rhd$ as follows. We place all candidates from~$A$ first in the agenda~$\rhd$, in arbitrary order, followed by~$p$. 
We maintain a set~$S$ of placed candidates, initially $S=A\cup \{p\}$. Next, candidates from $\inneighbor{G}{p}$ are added iteratively: in each iteration, let $S' = \outneighbor{G}{S}$, place candidates from $S'$ next in the agenda (in arbitrary order), and update $S \coloneqq S \cup S'$. Repeat until $S' = \emptyset$. 

After the iterations, if all candidates are placed,~$p$ wins and the algorithm returns ``\yes''. Otherwise, any remaining unplaced candidate beats~$p$ and is not beaten by any placed candidate. Hence,~$p$ cannot be the winner, and the algorithm returns ``\no''.

We now analyze the time complexity of the algorithm. By Lemma~\ref{lemma-majority-graph}, the majority graph~$G$ can be computed in time~$\bigo{n \cdot m^{1.69} + n^{0.55} \cdot m^2}$. The set~$A$ can be determined in time~$\bigo{m}$. Placing the candidates in $A \cup \{p\}$ into the agenda takes time~$\bigo{m}$, and placing the remaining candidates takes time~$\bigo{m^2}$.  
Overall, the algorithm runs in time 
$\bigo{n \cdot m^{1.69} + n^{0.55} \cdot m^2}$.
\end{proof}

We now present the result for {\prob{Coalition Manipulation}}.

\begin{theorem}
\label{thm-manipulation-p}
{\prob{Coalition Manipulation-TSMR}} can be solved in time $\bigo{n \cdot m^{1.69} + n^{0.55} \cdot m^2}$, where~$m$ is the number of candidates and~$n$ is the number of votes in the election.
\end{theorem}

\begin{proof}
Let $I=((C, V), p, \rhd, k)$ be an instance of {\prob{Coalition Manipulation-TSMR}}, with $m = \abs{C}$ and $n = \abs{V}$. If $k>n$, we immediately return ``\yes''; in this case,~$p$ can be made the TSMR winner by adding~$k$ votes ranking~$p$ first. 

We assume now that $k\leq n$. Let $B = C \setminus \{p\}$. Let~$V'$ be the multiset of $k$ votes with the same preference $p \succ \overrightarrow{B}$, where $\overrightarrow{B}$ is the linear order induced by $\rhd$ on $B$. We return ``\yes'' if and only if $p$ is the TSMR winner of $(C, V \cup V', \rhd)$.

Regarding the running time, constructing~$V'$ takes $\bigo{k \cdot m}$ time. By Lemma~\ref{lem-TSMR-winner-determination}, verifying whether~$p$ is the TSMR winner can be done in time $\bigo{n \cdot m^{1.69} + n^{0.55} \cdot m^2}$.

To prove correctness, we assume that $I$ is a {\yesins}, and show that $V'$ is a feasible solution of $I$. Observe that~$I$ has a feasible solution~$U$ where~$p$ is ranked first in all votes. If~$U$ contains a vote~$\succ$ with two  candidates~$x, y\in B$ such that~$y$ is ranked immediately before~$x$ in $\succ$ but $x \rhd y$, we swap their positions. 
Any candidate previously beaten by at least one of its predecessors remains so after the swap, and any candidate not previously beating~$p$ still does not beat~$p$. Consequently,~$p$ remains the TSMR winner after the swap. 
By exhaustively applying such swaps, we eventually transform~$U$ into~$V'$ without losing feasibility.
\end{proof}

\subsection{Constructive Controls}
\label{sec-control}
In this section, we study constructive control problems for TSMR.
We first present results for control by adding/deleting votes. We show that these problems are {\wbh} with respect to several meaningful parameters, both when the distinguished candidate is first in the agenda and when they are last.

\begin{theorem}
\label{thm-ccav-hard}
{\prob{CCAV-TSMR}} is {\memph{\wbh}} with respect to the number of added votes plus the number of registered votes, even when the distinguished candidate is first in the agenda.
\end{theorem}

\begin{proof}
We prove the theorem via a reduction from {\prob{RBDS}}. Let $I=(G, \kappa)$ be an instance of {\prob{RBDS}}, where~$G$ is a bipartite graph with the bipartition~$(\rs, \bs)$. We construct an instance of {\prob{CCAV-TSMR}} as follows. For each vertex in~$G$, we create a candidate with the same symbol. In addition, we introduce a candidate~$p$. Let $C = \bs \cup \rs \cup \{p\}$, and define the agenda as $\rhd = (p, \overrightarrow{\bs}, \overrightarrow{\rs})$. We create a multiset~$V$ of $\kappa+1$ registered votes: 
\begin{itemize}
    \item $\kappa$ votes with the preference $\overleftarrow{\bs}\Succ \overleftarrow{\rs}\Succ p$; and
    \item one vote with the preference $\overleftarrow{\rs}\Succ p\Succ \overleftarrow{\bs}$.
\end{itemize}
Next, for each $b \in \bs$, we create an unregistered vote~$\succ_b$ with the preference 
\[p\Succ \left(\overleftarrow{\rs}\setminus \neighbor{G}{b}\right) \Succ b \Succ \left(\overleftarrow{\rs}[\neighbor{G}{b}]\right) \Succ \left(\overleftarrow{\bs}\setminus \{b\}\right).\]
Let~$W$ be the set of these~$\abs{\bs}$ unregistered votes.
Setting $k=\kappa$, the resulting instance of {\prob{CCAV-TSMR}} is $g(I)=((C, V\cup W), p, \rhd, k)$.  

We now prove the correctness of this reduction, showing that~$I$ is a {\yesins} of {\prob{RBDS}} if and only if~$g(I)$ is a {\yesins} of {\prob{CCAV-TSMR}}.

$(\Rightarrow)$ Suppose that~$I$ is a {\yesins}, i.e., there exists $\bs'\subseteq \bs$ such that $\abs{\bs'}=\kappa$ and~$\bs'$ dominates~$\rs$ in~$G$. Let $W'=\{\succ_b\, \setmid b\in \bs'\}$ be the set of the~$\kappa$ unregistered votes corresponding to~$\bs'$, and define $\elec=(C, V\cup W')$. Note that $\abs{V \cup W'} = 2\kappa + 1$. We show that~$p$ is the TSMR winner of~$\elec$ with respect to~$\rhd$. 

For any $b \in \bs$, one registered vote and all the~$\kappa$ votes in~$W'$ rank~$p$ before~$b$, so~$p$ beats~$b$ in~$\elec$. Since~$p$ precedes all candidates from~$\bs$ in~$\rhd$, no candidate from~$\bs$ can win~$\elec$. 

Now consider a candidate $r \in \rs$. Since~$\bs'$ dominates~$\rs$, there exists $b \in \bs'$ with $r \in \neighbor{G}{b}$. By construction,~$b$ is ranked before~$r$ in $\succ_b$. Additionally, $\kappa$ registered votes rank $b$ before~$r$. Thus, in total, at least $\kappa+1$ votes in $V \cup W'$ rank $b$ before $r$. Since $b\rhd r$, it follows that $r$ cannot win $\elec$. Consequently, no candidate from $\rs$ can win  $\elec$. 

Hence,~$p$ is the TSMR winner of~$\elec$ with respect to~$\rhd$. Since $\abs{W'}=\abs{\bs'}=\kappa=k$,~$g(I)$ is a {\yesins}.

$(\Leftarrow)$ Suppose that~$g(I)$ is a {\yesins}, i.e., there exists $W'\subseteq W$ of at most~$\kappa$ votes such that~$p$ is the TSMR winner of $(C,V\cup W')$ with respect to~$\rhd$. Observe that~$W'$ must contain exactly~$\kappa$ votes; otherwise, some candidate in~$\bs$ would prevent~$p$ from winning. Next, every candidate in~$\rs$ beats~$p$ and also all their predecessors in~$\rs$ with respect to $V\cup W'$, regardless of which~$\kappa$ votes are included in~$W'$. Since~$p$ precedes all those in~$\rs$, every $r \in \rs$ must be beaten by some candidate in~$\bs$. Consequently, for each $r \in \rs$, there exists at least one vote in~$W'$ ranking some~$b\in \bs$ before~$r$. By construction, this vote must be~$\succ_b$ such that~$b$ dominates~$r$ in~$G$. Therefore, the set $\{b \in \bs\, \setmid\, \succ_b \in W'\}$ dominates~$\rs$, and hence~$I$ is a {\yesins}.
\end{proof}

We now examine the scenario where the distinguished candidate occupies the final position in the agenda. By Observation~\ref{obs-a}, the last candidate is the TSMR winner if and only if it is a weak Condorcet winner. The {\wahns} result for {\prob{CCAV-Condorcet}} by Liu~et~al.~\cite{DBLP:journals/tcs/LiuFZL09} can be adapted to show the same hardness result for \prob{CCAV-WeakCondorcet}, which asks whether a limited number of votes can be added to make the distinguished candidate a weak Condorcet winner. We strengthen the result by establishing a {\wbhns} reduction, thereby ruling out completeness for~{\wa}.

\begin{theorem}
\label{thm-ccav-hard-last}
{\prob{CCAV-TSMR}} is {\memph\wbh} parameterized by the total number of registered and added votes, even when the distinguished candidate is last in the agenda. 
\end{theorem}

\begin{proof}
We prove the theorem via a reduction from {\prob{RBDS}}. Let $I=(G, \kappa)$ be an instance of {\prob{RBDS}}, where~$G$ is a bipartite graph with the bipartition~$(\rs, \bs)$. We create an instance of {\prob{CCAV-TSMR}} as follows. The candidate set is $C=\rs\cup \{p, q\}$, and the agenda is $\rhd=(\overrightarrow{\rs}, q, p)$. We create a multiset~$V$ of~$\kappa$ registered votes:
\begin{itemize}
    \item $\kappa-1$ votes with the preference $q\Succ p\Succ \overrightarrow{\rs}$; and
    \item one vote with the preference $q\Succ \overrightarrow{\rs}\Succ p$.
\end{itemize}
For each $b\in \bs$, we create one unregistered vote~$\succ_b$ with the preference
\[\left(\overrightarrow{\rs}\setminus \neighbor{G}{b}\right) \Succ p \Succ \left(\overrightarrow{\rs}[\neighbor{G}{b}]\right) \Succ q.\] 
For a given $\bs'\subseteq \bs$, let~$W(\bs')=\{\succ_b\, \setmid b\in \bs'\}$ be the multiset of unregistered votes corresponding to~$\bs'$. Let $k=\kappa$. The instance of {\prob{CCAV-TSMR}} is $g(I)=((C, V\cup W(\bs)), p, \rhd, k)$. We now show the correctness of the reduction.

$(\Rightarrow)$ Assume that~$I$ is a {\yesins} of {\prob{RBDS}}, i.e., there exists $\bs'\subseteq \bs$ such that $\abs{\bs'}=\kappa$ and~$\bs'$ dominates~$\rs$ in~$G$. Let $\elec=(C, V\cup W(\bs'))$. Clearly,~$E$ contains exactly $\abs{V\cup W(\bs')}=2\kappa$ votes. We show that~$g(I)$ is a {\yesins} of {\prob{CCAV-TSMR}} by proving that~$p$ is the TSMR winner of~$\elec$ with respect to~$\rhd$. 
First, since all $\kappa$ registered votes rank~$q$ before~$p$, the candidates~$p$ and~$q$ are tied in~$\elec$. Second, since~$\bs'$ dominates~$\rs$, for every $r \in \rs$, there exists $b \in \bs'$ that dominates~$r$. By construction,~$p$ is ranked before~$r$ in the vote $\succ_b\, \in W(\bs')$ corresponding to~$b$. Together with the $\kappa-1$ registered votes from~$V$ that also rank~$p$ before~$r$, there are at least~$\kappa$ votes in~$E$ ranking~$p$ before~$r$. Therefore,~$p$ either beats or ties with~$r$ in~$\elec$. Because this holds for all candidates $r \in \rs$, and~$p$ is last in the agenda~$\rhd$, it follows that~$p$ is the TSMR winner of~$\elec$ with respect to~$\rhd$. 

$(\Leftarrow)$ Assume that $g(I)$ is a {\yesins} of  {\prob{CCAV-TSMR}}, i.e., there exists $\bs'\subseteq \bs$ such that $\abs{\bs'}\leq k=\kappa$ and~$p$ is the TSMR winner of $\elec=(C, V\cup W(\bs'))$ with respect to~$\rhd$. Since~$p$ is last in~$\rhd$, it cannot be beaten by any candidate in~$\elec$. It follows that $\abs{\bs'} = \kappa$, since otherwise~$q$ would beat~$p$. We have $\abs{V \cup W(\bs')} = 2\kappa$. Now consider any $r \in \rs$. As exactly $\kappa - 1$ registered votes rank~$p$ before~$r$, there exists at least one $b \in \bs'$ such that~$p$ is ranked before~$r$ in the vote~$\succ_b$. This implies that for each $r \in \rs$ there exists $b \in \bs'$ that dominates $r$, and hence $\bs'$ dominates $\rs$ in~$G$. Thus,~$I$ is a {\yesins} of {\prob{RBDS}}.
\end{proof}

We move on to constructive control by deleting votes. This problem admits two natural parameters: the maximum number~$k$ of votes that may be deleted and its dual~$n - k$, the minimum number of votes that must remain after deletion, where~$n$ is the total number of votes. We show that the problem is {\wbh} with respect to both parameters, even when the distinguished candidate appears either first or last in the agenda. These results are captured in Theorems~\ref{thm-ccdv-hard-deleted-first}--\ref{thm-ccdv-wbh-not-deleted-last}.

\begin{theorem}
\label{thm-ccdv-hard-deleted-first}
{\prob{CCDV-TSMR}} is {\memph\wbh} with respect to the number of deleted votes, even when the distinguished candidate appears first in the agenda.
\end{theorem}

\begin{proof}
We prove the theorem via a reduction from {\prob{RBDS}}. Let $I=(G, \kappa)$ be an instance of {\prob{RBDS}}, where~$G$ is a bipartite graph with the bipartition~$(\rs, \bs)$. We assume that~$G$ contains no isolated vertices, $\kappa \geq 4$, and every red vertex has degree~$\ell$, where $\ell \geq 1$. These assumptions do not affect the {\wbhns} of the problem.\footnote{The assumptions that~$G$ contains no isolated vertices and that $\kappa \geq 4$ are straightforward. If an instance does not satisfy the third assumption, we can obtain an equivalent instance as follows. Let~$\ell$ be the maximum degree of the vertices in~$\rs$. For each $r \in \rs$ with degree less than~$\ell$, we introduce new degree-$1$ vertices adjacent only to~$r$ until its degree equals~$\ell$. Assuming the original graph~$G$ has no isolated vertices, there exists an optimal solution (a subset $\bs' \subseteq \bs$ dominating~$\rs$ with minimum cardinality) in the modified instance that excludes all newly introduced vertices, ensuring the equivalence between the two instances.} We construct an instance of {\prob{CCDV-TSMR}} as follows. The candidate set is $C=\rs\cup \{p, q, q'\}$, and the agenda is $\rhd=(p,q',\overrightarrow{\rs},q)$. We create six groups of votes:
\begin{itemize}
    \item a multiset~$V_1$ of~$\ell+1$ votes with the preference $q'\Succ p\Succ q\Succ \overleftarrow{\rs}$;
    \item a multiset~$V_2$ of $\kappa+\ell-2$ votes with the preference $q\Succ p\Succ \overleftarrow{\rs}\Succ q'$;
    \item a multiset~$V_3$ of $\abs{\bs}-\kappa+1$ votes with the preference $\overleftarrow{\rs}\Succ p\Succ q\Succ q'$;
    \item a singleton~$V_4$ of one vote with the preference $\overleftarrow{\rs}\Succ q\Succ p\Succ q'$;
    \item a multiset~$V_5$ of $\kappa-2$ votes with the preference $\overleftarrow{\rs}\Succ q'\Succ p\Succ q$;
    \item for every blue vertex $b\in \bs$, one vote~$\succ_b$ with the preference
    $q\Succ q'\Succ \left(\overleftarrow{\rs}[\neighbor{G}{b}]\right) \Succ p \Succ \left(\overleftarrow{\rs}\setminus \neighbor{G}{b}\right)$.
\end{itemize}
Let~$V$ denote the multiset of the above $2\abs{\bs}+\kappa+2\ell-1$ votes. For a given $\bs'\subseteq \bs$, let $V(\bs')=\{\succ_b\, \setmid b\in \bs'\}$. Note that each $r \in \rs$ is ranked after~$p$ in exactly $\abs{\bs} - \ell$ votes from~$V(\bs)$.

Defining $k = \kappa$, the resulting instance of {\prob{CCDV-TSMR}} is $g(I) = ((C, V), p, \rhd, k)$, which can be constructed in polynomial time. Now we prove the correctness of the reduction. 

$(\Rightarrow)$ Assume that~$I$ is a {\yesins} of {\prob{RBDS}}, i.e., there exists $\bs'\subseteq \bs$ with $\abs{\bs'}=\kappa$  such that~$\bs'$ dominates~$\rs$ in $G$. Let $\elec=(C, V\setminus V(\bs'))$. To prove that~$g(I)$ is a {\yesins}, we show that~$p$ is the TSMR winner of~$\elec$ with respect to the agenda~$\rhd$. 

It suffices to show that~$p$ beats every other candidate in~$\elec$. Let $r\in \rs$. Since~$\bs'$ dominates~$\rs$, there exists at least one vertex $b\in \bs'$ that dominates~$r$. By construction, the vote~$\succ_b$ ranks~$r$ before~$p$. Since exactly $\abs{\bs}-\ell$ votes in~$V(\bs)$ rank~$p$ before~$r$, $\abs{V(\bs')}=\abs{\bs'}=\kappa$, and $\succ_b\, \in V(\bs')$, at least $(\abs{\bs}-\ell)-(\kappa-1)$ votes in $V(\bs \setminus\bs')$ rank~$p$ before~$r$. Additionally, all votes in $V_1\cup V_2$ rank~$p$ before~$r$. Thus, at least $(\abs{\bs}-\ell)-(\kappa-1)+(\ell+1)+(\kappa+\ell-2)=\abs{\bs}+\ell$ votes in~$\elec$ rank~$p$ before~$r$. Since $\abs{V\setminus V(\bs')}=2\abs{\bs}+2\ell-1$, we know that~$p$ beats~$r$ in~$\elec$. One can verify that  $\abs{\bs}+\ell$ votes rank~$p$ before~$q$ and~$q'$ in $V\setminus V(\bs')$, implying that~$p$ also beats~$q$ and~$q'$ in~$\elec$.

$(\Leftarrow)$ Assume that $g(I)$ is a {\yesins}, i.e., there exists $V'\subseteq V$ with $\abs{V'}\leq k=\kappa$ such that~$p$ is the TSMR winner of $\elec=(C, V\setminus V')$ with respect to~$\rhd$. By construction, regardless of which votes are contained in~$V'$, every candidate from $C\setminus \{p\}$ beats all her predecessors in $C\setminus \{p\}$. Since~$p$ is first in the agenda and wins~$\elec$, it follows that~$p$ beats all other candidates. This implies that $V'$ is disjoint from $V_1 \cup V_3 \cup V_5$ and that $\abs{V'} = \kappa$, since otherwise~$p$ could not beat~$q$ in~$\elec$. Similarly,~$V'$ must also be disjoint from $V_2 \cup V_4$, since otherwise~$p$ would not beat~$q'$. Consequently, we have $V' \subseteq V(\bs)$. Let $\bs'\subseteq \bs$ be such that $V(\bs')=V'$. 

We now show that~$\bs'$ dominates~$\rs$ in~$G$. Suppose, for contradiction, that there exists $r\in \rs$ not dominated by any vertex in~$\bs'$. By construction, all votes in~$V(\bs')$ rank~$p$ before~$r$. 
Among the votes in $V(\bs)$, exactly $\abs{\bs} - \ell$ rank~$p$ before~$r$, and since $\abs{V(\bs')} = \kappa$, there remain $(\abs{\bs} - \ell) - \kappa$ such votes in $V(\bs) \setminus V'$. 
Additionally, all votes in $V_1 \cup V_2$ rank~$p$ before~$r$. Thus, there are $(\abs{\bs}-\ell)-\kappa+\abs{V_1\cup V_2}=\abs{\bs}+\ell-1$ votes in~$\elec$ ranking~$p$ before~$r$. Since $\abs{V\setminus V'}=2\abs{\bs}+2\ell-1$, candidate~$p$ is beaten by~$r$ in~$\elec$, contradicting that~$p$ is the TSMR winner. Hence,~$\bs'$ dominates~$\rs$ in~$G$. Since $\abs{\bs'} = \kappa$, it follows that~$I$ is a {\yesins} of {\prob{RBDS}}.
\end{proof}

\begin{theorem}[$\bigstar$]
\label{thm-ccdv-hard-first}
{\prob{CCDV-TSMR}} is {\memph\wbh} with respect to the number of  votes not deleted, even when the distinguished candidate is first in the agenda.
\end{theorem}

\begin{theorem}[$\bigstar$]
\label{thm-ccdv-hard}
{\prob{CCDV-TSMR}} is {\emph{\wbh}} with respect to the number of deleted votes, even when the distinguished candidate is last in the agenda.
\end{theorem}

By Observation~\ref{obs-a}, Theorem~\ref{thm-ccdv-hard} strengthens the {\wahns} of {\prob{CCDV-WeakCondorcet}} established by Liu~et~al.~\cite{DBLP:journals/tcs/LiuFZL09}, which asks whether a limited number of votes can be deleted to make a distinguished candidate a weak Condorcet winner.

\begin{theorem}[$\bigstar$]
\label{thm-ccdv-wbh-not-deleted-last}
{\prob{CCDV-TSMR}} is {\memph\wbh} with respect to the number of votes not deleted, even when the distinguished candidate is last in the agenda.
\end{theorem}

We now study constructive control by adding or deleting candidates. In contrast to voter control, we establish only a single hardness result, presented in the following theorem.

\begin{theorem}
\label{thm-ccac-hard}
{\prob{CCAC-TSMR}} is {\memph\wbh} with respect to the number of added candidates. This holds even when the distinguished candidate is first in the agenda. 
\end{theorem}

\begin{proof}
We prove the theorem via a reduction from {\prob{RBDS}}. Let $I=(G, \kappa)$ be an instance of {\prob{RBDS}}, where~$G$ is a bipartite graph with the bipartition~$(\rs, \bs)$. We construct an instance $g(I)$ of {\prob{CCAC-TSMR}}  as follows. For each vertex in~$G$, we create a corresponding candidate, using the same symbol for notational simplicity. Additionally, we introduce a distinguished candidate~$p$. Let $C = \rs \cup \{p\}$ and $D = \bs$. Set $k = \kappa$, and define the agenda as~$\rhd = (p, \overrightarrow{\bs}, \overrightarrow{\rs})$. 

We then construct a set~$V$ of votes on $C\cup D$ satisfying:
\begin{enumerate}
\item[(i)] every candidate in~$\rs$ beats all its predecessors in~$\rs\cup \{p\}$;
\item[(ii)] $p$ beats every candidate in~$\bs$; and
\item[(iii)] for each $r\in \rs$ and each $b\in \bs$, if~$b$ dominates~$r$ in~$G$, then~$b$ beats~$r$; otherwise,~$r$ beats~$b$.
\end{enumerate}

By McGarvey’s theorem~\cite{McGarvey1953}, such votes can be constructed in polynomial time. The resulting instance is $g(I)=((C \cup D, V), p, \rhd, k)$. 
We now prove the correctness of the reduction. 

$(\Rightarrow)$ Assume that $I$ is a {\yesins} of {\prob{RBDS}}, i.e., there exists a subset $\bs'\subseteq \bs$ of~$\kappa$ vertices that dominate~$\rs$ in~$G$. Consider the election $\elec=(C \cup \bs', V)$. Let~$\rhd'$ be the agenda~$\rhd$ restricted to $C \cup \bs'$. Consider any~$r\in \rs$. Let~$b$ be a vertex in~$\bs'$ that dominates~$r$ in~$G$. By Condition (iii),~$b$ beats~$r$ in~$\elec$, and since~$b$ precedes~$r$ in~$\rhd'$,~$r$ cannot be the TSMR winner of~$\elec$ with respect to~$\rhd'$.  
By Condition~(ii), every candidate in~$\bs'$ is beaten by~$p$. Since~$p$ precedes candidates in~$\bs'$, no candidate from~$\bs'$ can be the TSMR winner of~$\elec$. Hence,~$p$ is the TSMR winner of~$\elec$, and~$g(I)$ is a {\yesins}. 

$(\Leftarrow)$ Assume that~$I$ is a {\noins} of {\prob{RBDS}}. Consider any subset $\bs' \subseteq \bs$ of size at most~$\kappa$. Let $\elec = (C \cup \bs', V)$, and let~$\rhd'$ be the agenda~$\rhd$ restricted to~$C \cup \bs'$. Let $r \in \rs$ be a vertex not dominated by any vertex in~$\bs'$. By Condition~(i), $r$ beats all its predecessors in~$\rs\cup \{p\}$ in~$\elec$, and by Condition (iii),~$r$ beats all candidates in~$\bs'$. Therefore,~$r$ beats all its predecessors in~$\elec$. Hence,~$r$ is a source in the spanning subdigraph of the majority graph of~$\elec$ induced by the forward arcs with respect to~$\rhd'$. Since~$p$ precedes~$r$ in~$\rhd'$,~$p$ cannot be the TSMR winner of~$\elec$.  
As this holds for all $\bs' \subseteq \bs$ of size at most~$\kappa$, we conclude that~$g(I)$ is a {\noins}.
\end{proof}

For the case where the distinguished candidate appears last in the agenda, the following corollary follows from Observation~\ref{obs-a} and the fact that weak Condorcet winners are immune to \prob{CCAC}~\cite{Bartholdi92howhard}.

\begin{corollary}
\label{cor-ccac-last-immne}
If the distinguished candidate is last in the agenda, {\memph{TSMR}} is immune to {\prob{CCAC}}.
\end{corollary}

For {\prob{CCDC-TSMR}}, a greedy polynomial-time algorithm can be obtained.

\begin{theorem}
\label{thm-ccdc-p}
\textsc{CCDC-TSMR} is solvable in time $\bigo{n \cdot m^{1.69} + n^{0.55} \cdot m^2}$, where~$m$ denotes the number of candidates, and~$n$ denotes the number of votes in the election.
\end{theorem}

\begin{proof}
Let $((C, V), p, \rhd, k)$ be an instance of \textsc{CCDC-TSMR}, with $m = \abs{C}$ and $n = \abs{V}$. 
To solve this instance, we proceed as follows: 
First, we remove all predecessors of~$p$ in the agenda~$\rhd$ that beat~$p$ with respect to~$V$. Second, we iteratively remove each successor~$c$ of~$p$ such that~$c$ is unbeaten by any of its current predecessors in the remaining agenda. 

These steps rely on constructing the majority graph of $(C, V)$, which, by Lemma~\ref{lemma-majority-graph}, can be done in time $\bigo{n \cdot m^{1.69} + n^{0.55} \cdot m^2}$. 
After these removals, candidate~$p$ becomes the TSMR winner. We return ``\yes'' if and only if  at most~$k$ candidates have been removed. 
The correctness of the algorithm is straightforward. The total running time is dominated by the construction of the majority graph. 
\end{proof}

\subsection{Destructive Controls}
We now turn to destructive control problems.  
It is known that {\prob{DCAV-WeakCondorcet}} and {\prob{DCDV-WeakCondorcet}} are solvable in polynomial time~\cite{DBLP:journals/ai/HemaspaandraHR07}. These problems ask whether a limited number of votes can be added (\prob{DCAV}) or deleted (\prob{DCDV}) to prevent a distinguished candidate from being a weak Condorcet winner. More precisely, they can be solved in time~$\bigo{m \cdot n}$, where~$m$ is the number of candidates and~$n$ is the total number of votes (both registered and unregistered) in the election. 
The main idea is to identify a candidate who beats the distinguished candidate in the final election. There are at most 
$m-1$ such candidates to consider. For each fixed choice, we greedily add or delete the required number of votes. 
Combining Observation~\ref{obs-a} with this result, we obtain the following corollary. 

\begin{corollary}[\cite{DBLP:journals/ai/HemaspaandraHR07}]
\label{cor-dcav-dcdv-p-last}
{\memph{\prob{DCAV-TSMR}}} and {\memph{\prob{DCDV-TSMR}}} are solvable in time~$\bigo{m \cdot n}$ 
when the distinguished candidate appears last in the agenda, where~$m$ is the number of candidates and~$n$ is the number of votes in the election.
\end{corollary}

Next, we show that if the distinguished candidate does not appear last in the agenda, both {\prob{DCAV-TSMR}} and {\prob{DCDV-TSMR}} become computationally hard. 
This holds whenever the distinguished candidate occupies any fixed position in the agenda other than the last. 

\begin{theorem}
\label{thm-dcav-hard-first}
{\prob{DCAV-TSMR}} is {\memph{\wbh}} with respect to the number of added votes plus the number of registered votes. This holds as long as the distinguished candidate is not last in the agenda.
\end{theorem}

\begin{proof}
We prove the theorem via a reduction from {\prob{RBDS}}. Let $(G, \kappa)$ be an instance of {\prob{RBDS}}, where~$G$ is a bipartite graph with the bipartition~$(\rs,\bs)$. We construct an instance of {\prob{DCAV-TSMR}} as follows. Let $C = \rs \cup \{p, q\}$, and let~$\rhd$ be any agenda on~$C$ where~$q$ is last. 
We define a multiset~$V$ of $\kappa + 2$ registered votes as follows: 
\begin{itemize}
    \item $\kappa-1$ votes with the preference $p\Succ q\Succ \overrightarrow{\rs}$;
    \item two votes with the preference $p\Succ \overrightarrow{\rs}\Succ q$; and 
    \item one vote with the preference $q\Succ p\Succ \overrightarrow{\rs}$.
\end{itemize}
The unregistered votes are constructed based on~$\bs$. For each $b \in \bs$, we construct a vote~$\succ_b$ with the preference
\[\left(\overrightarrow{\rs}\setminus \neighbor{G}{b}\right) \Succ q \Succ p \Succ \left(\overrightarrow{\rs}[\neighbor{G}{b}]\right).\]
For a given~$\bs'\subseteq \bs$, let~$W(\bs')=\{\succ_b\, \setmid b\in \bs'\}$ be the multiset of unregistered votes corresponding to~$\bs'$.
Let~$W=W(\bs)$. 
Setting $k=\kappa$, the resulting instance is $g(I)=((C, V\cup W), p, \rhd, k)$. We prove the correctness of the reduction as follows.

$(\Rightarrow)$ Suppose that $I$ is a {\yesins} of {\prob{RBDS}}, i.e., there exists $\bs'\subseteq \bs$ of~$\kappa$ vertices that dominate~$\rs$ in~$G$. Then, in the election $(C, V\cup W(\bs'))$,~$q$ beats or ties with every other candidate, implying that~$q$ is the TSMR winner. Thus,~$g(I)$ is a {\yesins} of {\prob{DCAV-TSMR}}.

$(\Leftarrow)$ Suppose that~$g(I)$ is a {\yesins}, i.e., there exists $W'\subseteq W$ of at most~$k$ votes such that~$p$ is not the TSMR winner of $\elec=(C, V\cup W')$ with respect to~$\rhd$. Regardless of the choice of $W'$,~$p$ beats every candidate in~$\rs$, leaving $q$ as the only candidate that can prevent~$p$ from winning. Since~$q$ is last in the agenda, it wins exactly when it beats or ties with every other candidate. It follows that~$W'$ contain exactly~$\kappa$ votes, since otherwise~$p$ beats~$q$ in~$\elec$. Moreover, for each $r\in \rs$, at least one vote in~$W'$ ranks~$q$ before~$r$. By construction, an unregistered vote~$\succ_b$ ranks~$q$ before~$r$ if and only if~$b$ dominates~$r$ in~$G$. Hence, the set of vertices corresponding to~$W'$ dominates~$\rs$, and~$I$ is a {\yesins} of {\prob{RBDS}}.
\end{proof}

By slightly modifying the reductions used in the proofs of Theorems~\ref{thm-ccdv-hard} and~\ref{thm-ccdv-wbh-not-deleted-last}, we obtain the following parameterized hardness results for {\prob{DCDV-TSMR}}, respectively.

\begin{theorem}[$\bigstar$]
\label{thm-dcdv-hard-deleted-first}
{\prob{DCDV-TSMR}} is {\memph{\wbh}} with respect to the number of deleted votes, as long as the distinguished candidate is not last in the agenda.
\end{theorem}

\begin{theorem}[$\bigstar$]
\label{thm-dcdv-wbh-not-deleted-first}
{\prob{DCDV-TSMR}} is {\memph\wbh} with respect to the number of votes not deleted, as long as the distinguished candidate is not last in the agenda.
\end{theorem}

For destructive control by modifying candidates, polynomial-time solvability results hold true, irrespective of the position of the distinguished candidate in the agenda.

\begin{theorem}
\label{thm-dcac-p}
{\prob{DCAC-TSMR}}  can be solved in time $\bigo{n \cdot m}$, where~$m$ and~$n$ represent the number of candidates and the number of votes in the election, respectively.
\end{theorem}

\begin{proof}
Let $I = ((C \cup D, V), p, \rhd, k)$ be an instance of {\prob{DCAC-TSMR}} with $m=\abs{C\cup D}$ and $n=\abs{V}$. 
We assume $k \geq 1$ and that~$p$ is the TSMR winner of~$(C, V)$; otherwise,~$I$ can be solved trivially.  

Since~$p$ wins $(C, V)$, it is not beaten by any of its predecessors from~$C$ in the agenda~$\rhd$, and each successor~$c \in C \setminus \{p\}$ is beaten by at least one of its own predecessors. If~$p$ has a predecessor $c \in D$ such that~$c$ beats~$p$, we conclude that~$I$ is a {\yesins}, because~$p$ would not win $(C \cup \{c\}, V)$. This step can be implemented in time $\bigo{n\cdot \abs{D}}=\bigo{n\cdot m}$. 
Additionally, if~$p$ has a successor $c \in D$ in~$\rhd$ unbeaten by any of its predecessors in~$C$, then~$I$ is also a {\yesins}. This check also runs in time~$\bigo{n \cdot m}$.  
If neither of the conditions holds, then~$p$ remains the winner regardless of which subset of~$D$ is added, and we return ``\no''.
\end{proof}

The result for \prob{DCDC-TSMR} follows from Theorem~\ref{thm-ccdc-p}, as an instance is a \yesins{} if and only if there exists a candidate $c \in C \setminus \{p\}$ who can be made the TSMR winner via the deletion of at most $k$ candidates. Thus, solving \prob{DCDC-TSMR} reduces to solving at most $m - 1$ instances of \prob{CCDC-TSMR}.

\begin{corollary}
\label{cor-dcdc-p}
{\prob{DCDC-TSMR}} can be solved in time $\bigo{n \cdot m^{2.69} + n^{0.55} \cdot m^3}$, where $m$ is the number of candidates and $n$ the number of votes.
\end{corollary}

Note that when the distinguished candidate is last in the agenda, TSMR is immune to DCDC, which follows directly from the definition of TSMR.

\section{Possible and Necessary Winners}
\label{sec-possible-necessary}
This section studies the {\prob{Possible Winner}} and the {\prob{Necessary Winner}} problems. 
Bredereck~et~al.~\cite{DBLP:journals/jair/BredereckCNW17} showed that, except for {\prob{Necessary Winner-Successive}}, which is polynomial-time solvable, other cases for the amendment rule and the successive rule are computationally hard ({\nphns} for {\prob{Possible Winner}} and {\conphns} for {\prob{Necessary Winner}}). We prove TSMR matches successive rule complexity for these problems. 

We first present the result for \prob{Possible Winner}, then \prob{Necessary Winner}, and finally analyze two parameters measuring input incompleteness.

\subsection{The {\prob{Possible Winner}} Problem}
\label{sec-possible}

\begin{theorem}
\label{thm-possible-hard}
{\prob{Possible Winner-TSMR}} is {\memph{\nph}}, even when the agenda is complete and the distinguished candidate is first. 
\end{theorem}

\begin{proof}
We prove the theorem via a reduction from {\prob{RBDS}}. Let $(G, \kappa)$ be an instance of {\prob{RBDS}}, where~$G$ is a bipartite graph with the partition $(\rs, \bs)$. As in the proof of Theorem~\ref{thm-ccdv-hard-deleted-first}, we assume that~$G$ does not contain any isolated vertices, and all vertices in~$\rs$ have the same degree~$\ell$ where~$\ell\geq 1$.
We create an instance of {\prob{Possible Winner-TSMR}} as follows. 
Let $C=\rs\cup \{p, q\}$ and let  $\rhd=(p, q, \overrightarrow{\rs})$. 
We create five groups of votes, where only the first group contains partial votes:
\begin{itemize}
\item For each $b\in \bs$, one partial vote~$\succ_b$ with the preference
\[\left(\overleftarrow{\rs}[\neighbor{G}{b}]\right) \Succ p \Succ \left(\overleftarrow{\rs}\setminus \neighbor{G}{b}\right)\] and 
\[q\Succ \left(\overleftarrow{\rs}\setminus \neighbor{G}{b}\right);\]

\item A multiset~$V_1$ of~$\abs{\bs}$ complete votes with the preference $\overleftarrow{\rs} \Succ q\Succ p$;

\item A multiset~$V_2$ of $2\ell+\kappa$ complete votes with the preference $q\Succ \overleftarrow{\rs}\Succ p$;

\item A multiset~$V_3$ of $\ell+2\kappa+1$ complete votes with the preference $\overleftarrow{\rs}\Succ p\Succ q$;

\item A multiset~$V_4$ of $\ell+\kappa$ complete votes with the preference $p\Succ q\Succ \overleftarrow{\rs}$.
\end{itemize}
Let~$V({\bs})=\{\succ_b\, \setmid b\in \bs\}$ be the set of the~$\abs{\bs}$ partial votes in the first group. Let~$V$ be the multiset of the above $2\abs{\bs}+4\ell+4\kappa+1$ votes, and define $V({\overline{\bs}})=V\setminus V(\bs)$.

The instance of {\prob{Possible Winner-TSMR}} is $((C, V), p, \rhd)$, which can be constructed in polynomial time.
We now prove the correctness of the reduction.

$(\Rightarrow)$ Suppose that there is a subset $\bs'\subseteq \bs$ such that $\abs{\bs'}=\kappa$ and~$\bs'$ dominates~$\rs$ in~$G$. We complete each partial vote~$\succ_b$, where $b\in \bs$, as follows:
\begin{itemize}
    \item if $b\in \bs'$, complete it as
    $q\Succ \left(\overleftarrow{\rs}[\neighbor{G}{b}]\right) \Succ p \Succ \left(\overleftarrow{\rs}\setminus \neighbor{G}{b}\right)$,
    \item otherwise, complete it as
    $\left(\overleftarrow{\rs}[\neighbor{G}{b}]\right) \Succ p \Succ q \Succ \left(\overleftarrow{\rs}\setminus \neighbor{G}{b}\right)$.
\end{itemize}
With respect to this completion,~$p$ beats~$q$, and~$q$ beats all candidates in~$\rs$.
Then, by the agenda definition,~$p$ is the TSMR winner of the completion of~$(C, V)$.

$(\Leftarrow)$ Suppose that there is a completion~$V'$ of~$V(\bs)$ such that~$p$ wins $\elec=(C, V(\overline{\bs})\cup V')$ with respect to~$\rhd$. Let~$\bs'$ be the subset of~$\bs$ corresponding to votes in~$V'$ ranking~$p$ before~$q$, and let $\bs''=B\setminus \bs'$. 
In any completion of~$(C, V)$, every candidate in~$\rs$ beats all its predecessors in $\rs\cup \{p\}$. Since $p$ wins $\elec$, it follows that (1)~$q$ beats all candidates in~$\rs$, and 
(2)~$q$ is beaten by~$p$ in~$\elec$. 
As~$V(\overline{\bs})$ contains exactly $2\ell+3\kappa+1$ votes (from $V_3\cup V_4$) ranking~$p$ before~$q$, Condition~(2) implies that $\abs{\bs'}\geq \abs{\bs}-\kappa$, and thus $\abs{\bs''}\leq \kappa$. 

We prove that~$\bs''$ dominates~$\rs$ in~$G$. Suppose, for contradiction, that there exists $r\in \rs$ not dominated by any vertex in~$\bs''$. Equivalently, all the~$\ell$ neighbors of~$r$ in~$G$ are contained in~$\bs'$. 
Consequently, at least~$\ell$ votes in~$V'$ rank~$r$ before~$q$.
Combined with the votes in $V_1\cup V_3$, there are at least $\abs{\bs}+2\ell+2\kappa+1$ votes ranking~$r$ before~$q$, implying that~$r$ beats~$q$ in~$\elec$. 
This contradicts Condition~(1). Thus, $\bs''$ dominates $\rs$. 
Since~$\abs{\bs''}\leq \kappa$, it follows that $(G, \kappa)$ is a {\yesins} of {\prob{RBDS}}.
\end{proof}

Our reduction in the proof of Theorem~\ref{thm-possible-hard} differs from those in~\cite{DBLP:journals/jair/BredereckCNW17}, which show the {\nphns} of {\prob{Possible Winner-Amendment}} and {\prob{Possible Winner-Successive}} for complete agendas via reductions from {\prob{Independent Set}} and {\prob{Vertex Cover}}. 
Therein, the distinguished candidate~$p$ occupies the penultimate position of the agenda in the former and the third position in the latter. 
Our reduction can be adapted to prove the {\nphns} of {\prob{Possible Winner-TSMR}} when~$p$ occupies the $i$-th position for any constant $i$, by adding $i-1$ dummy candidates before~$p$ in the agenda and ranking them after all other candidates in all votes. 

In contrast, {\prob{Possible Winner-TSMR}} is polynomial-time solvable if the given agenda is complete and~$p$ appears last. This follows from Observation~\ref{obs-a} and the polynomial-time solvability of determining if a partial election can be completed such that a candidate becomes a (weak) Condorcet winner~\cite{Konczak05votingprocedures}.{\footnote{The result of~\cite{Konczak05votingprocedures} concerns Condorcet winners, but the algorithm also applies to weak Condorcet winners.}} 

\begin{corollary}
    \label{cor-possible-p-last}
{\prob{Possible Winner-TSMR}} can be solved in time $\bigo{n\cdot m}$ when the agenda is complete and the distinguished candidate appears last, where~$m$ and~$n$ represent the number of candidates and the number of votes in the election, respectively.
\end{corollary}

The algorithm in~\cite{Konczak05votingprocedures} can be trivially adapted to show the polynomial-time solvability of {\prob{Possible Winner-Amendment}} for complete agendas with $p$ among the top-$2$ positions.  
This highlights a radical complexity shift for the amendment rule when the distinguished candidate moves from the second to the third position in the agenda.
Our next result also reveals a complexity shift for TSMR when~$p$ moves up by just one position from the last.

\begin{theorem}[$\bigstar$]
\label{thm-possible-hard-second-last}
{\prob{Possible Winner-TSMR}} is {\memph{\nph}}, even for complete agendas with the distinguished candidate in the penultimate position. 
\end{theorem}

It remains open whether a similar complexity shift occurs for the successive rule. This entails determining the complexity of {\prob{Possible Winner-Successive}} when the agenda is complete and the distinguished candidate appears last.

\subsection{The {\prob{Necessary Winner}} Problem}

In contrast to the hardness of {\prob{Possible Winner-TSMR}}, we show that {\prob{Necessary Winner-TSMR}} is polynomial-time solvable. By definition, a candidate~$p$ is not a necessary winner of a partial election if and only if there exists a completion where~$p$ fails to win. We show that  the existence of such a completion can be efficiently checked.

\begin{theorem}
\label{thm-necessary-p}
{\prob{Necessary Winner-TSMR}} can be solved in time $\bigo{m^2\cdot n}$, where~$m$ and~$n$ denote the number of candidates and the number of votes in the election, respectively. 
\end{theorem}

\begin{proof}
Let $((C, V), p, \rhd)$ be an instance of {\prob{Necessary Winner-TSMR}}, with $m = \abs{C}$ and $n = \abs{V}$. 
We aim to determine whether a completion of $(C, V)$ and $\rhd$ exists in which~$p$ is not the TSMR winner. 
The distinguished candidate~$p$ fails to be the TSMR winner if and only if:
\begin{enumerate}
    \item[(1)] a predecessor of~$p$ beats~$p$, or
    \item[(2)] a successor of~$p$ is unbeaten by any of its predecessors.
\end{enumerate}
Our algorithm contains the following steps.
\begin{itemize}
    \item Step~1. \hfill
    
    We verify if Condition~(1) can be satisfied by at least one completion. 
    Let $B = \{c \in C \setminus \{p\} \setmid (p, c) \notin \rhd\}$ be the set of candidates that can appear as predecessors of~$p$ in some completion of~$\rhd$. This set can be identified in time~$\bigo{m}$. 
    If $B=\emptyset$, Condition~(1) cannot be satisfied, and we move to Step~2. 
Otherwise, for each $c\in B$, we apply a greedy strategy to determine whether $(C,V)$ admits a completion in which~$c$ beats~$p$. Let~$n'$ be the number of partial votes in~$V$ not ranking~$p$ before~$c$, which can be computed in time~$\bigo{n}$. 
If $n'>n/2$, then such a completion exists and the algorithm returns ``\no''. 
This step runs in time $\bigo{m}+\abs{B}\cdot \bigo{n}=\bigo{m \cdot n}$.

\item Step~2. \hfill 

Let $B' = \{c \in C \setminus \{p\} \setmid (c, p) \notin \rhd\}$ denote the set of candidates that could be successors of~$p$ in some completion of~$\rhd$.  This set can be computed in time~$\bigo{m}$.  
If $B'=\emptyset$, Condition~(2) cannot be satisfied by any completion, and the algorithm returns ``\yes''.  
Otherwise, for each $c \in B'$, we determine whether there exists a completion of $(C, V)$ and $\rhd$ that satisfies Condition~(2). Let $A_c = \{c' \in C \setmid (c', c) \in \rhd\}$ be the set of candidates that must precede~$c$. The algorithm returns ``{\no}'' if there exists a $c \in B'$ such that for every $c' \in A_c \cup \{p\}$, at most~$n/2$ partial votes in~$V$ rank~$c'$ before~$c$. In such a case, there exists a completion of~$\rhd$ where only candidates in $A_c \cup \{p\}$ are predecessors of $c$, and a completion of $(C, V)$ where~$c$ is unbeaten by any candidate in $A_c \cup \{p\}$. This completion satisfies Condition~(2), proving that $p$ is not a necessary winner.  
This check for a fixed $c$ takes $\bigo{m \cdot n}$ time. Iterating over all candidates in $B'$ yields a total running time of $(m-1) \cdot \bigo{m \cdot n} = \bigo{m^2 \cdot n}$.
\end{itemize}

If no candidate triggers a ``\no'' return in Steps~1 or~2, the algorithm returns ``\yes''.

Overall, the algorithm runs in time $\bigo{m^2 \cdot n}$.
\end{proof}

\subsection{Parameterizations of Vote Incompleteness}
\label{sec-remarks-incompleteness}
In the context of voting with partial information, the parameterized complexity of various measures of vote completeness has been extensively studied. A pair of candidates in a partial vote is undetermined if their relative preference is not specified. Two natural parameters are: 
\begin{enumerate}
    \item[(1)] the maximum number of undetermined pairs per vote, and 
    \item[(2)] the total number of undetermined pairs across all votes.
\end{enumerate} 

For many voting rules for which the {\prob{Possible Winner}} and {\prob{Necessary Winner}} problems are computationally hard, the hardness persists even when each partial vote contains only a constant number of undetermined pairs. However, the problems become fixed-parameter tractable ({\fpt}) when parameterized by the total number of undetermined pairs~\cite{DBLP:conf/ijcai/BetzlerHN09,DBLP:journals/jair/XiaC11}.

Our reduction in the proof of Theorem~\ref{thm-possible-hard} shows that {\prob{Possible Winner-TSMR}} remains {\nph} even for complete agendas where the distinguished candidate occupies the first position and each vote contains at most four undetermined pairs. This follows from the fact that {\prob{RBDS}} remains {\nph} even when each blue vertex has at most three neighbors~\cite{garey}. In such restricted instances, for each partial vote~$\succ_b$ introduced for a blue vertex $b \in \bs$, the pair $\{p, q\}$ and each $\{q, r\}$, where~$r$ is a neighbor of~$b$ in~$G$, are undetermined.

Bredereck~et~al.~\cite{DBLP:journals/jair/BredereckCNW17} implicitly showed that {\prob{Possible Winner-Successive}} remains {\nph}, even when the agenda is complete, the distinguished candidate appears in the penultimate position, and each vote contains at most four undetermined pairs. Additionally, {\prob{Possible Winner-Amendment}} is {\nph} even when the agenda is complete, the distinguished candidate appears third in the agenda, and each vote contains at most eight undetermined pairs. The {\conphns} of {\prob{Necessary Winner-Amendment}} also holds under the restrictions that the agenda is complete, the distinguished candidate appears last in the agenda, and each vote contains at most eight undetermined pairs.

With respect to the total number of undetermined pairs (denoted $\kappa$), a branching-based {\fpt} algorithm exists~\cite{DBLP:conf/ijcai/BetzlerHN09}: for each undetermined pair, the algorithm branches on the two possible preference orders. Its running time is $\bigos{2^\kappa}$, where $\bigos{\cdot}$ hides polynomial factors. Betzler, Hemmann, and Niedermeier~\cite{DBLP:conf/ijcai/BetzlerHN09} further proposed more efficient {\fpt} algorithms, as well as polynomial-time algorithms for special cases in which the undetermined pairs follow specific combinatorial patterns. Although these refinements lie beyond the scope of the present work, they provide a promising avenue for future research. 

\section{Other Algorithmic Lower Bounds}
\label{sec-lowerbounds}
While our principal aim in this paper is to delineate the parameterized complexity of strategic problems under TSMR, the reductions outlined in the preceding section carry additional implications for kernelization, approximability, and exact algorithmic lower bounds. In this section, we briefly discuss these ramifications.   
We use the following notations:
\begin{itemize}
\item $m_{\text{r}}$: number of registered candidates, i.e., $m_{\text{r}}=\abs{C}$.
\item $m_{\text{u}}$: number of unregistered candidates, i.e., $m_{\text{u}}=\abs{D}$.
\item $m$: total number of candidates.
\item $n_{\text{r}}$: number of registered votes, i.e., $n_{\text{r}}=\abs{V}$.
\item $n_{\text{u}}$: number of unregistered votes, i.e., $n_{\text{u}}=\abs{W}$.
\item $n$: total number of votes.
\item $k$: number of votes/candidates allowed to be added/deleted.
\end{itemize}

The algorithmic lower bounds in this section rely on standard complexity-theoretic assumptions. Readers unfamiliar with these assumptions may safely skip the technical details, noting that they are widely accepted in the literature.

\subsection{Lower Bounds for Kernelization}
\begin{definition}[Kernelization]
Let~$P$ be a parameterized problem. A kernelization algorithm (kernelization) for~$P$ is an algorithm which takes as input an instance $(X, \kappa)$ of~$P$ and outputs another instance $(X',\kappa')$ of~$P$ such that the following hold: 
\begin{itemize}
\item The algorithm runs in polynomial time. 
\item $(X, \kappa)$ is a {\yesins} of~$P$ if and only if $(X', \kappa')$ is a {\yesins} of~$P$.
\item $\abs{X'}\leq f(\kappa)$ for some computable function~$f$ of~$\kappa$.
\item $\kappa'\leq g(\kappa)$ for some computable function~$g$ of $\kappa$.
\end{itemize}
\end{definition}
The instance $(X', \kappa')$ returned by a kernelization algorithm is called a kernel of~$P$.  
If~$f(\kappa)$ is a polynomial function, we call $(X', \kappa')$ a polynomial kernel, and say that the problem~$P$ admits a polynomial kernel.

It is known that {\prob{RBDS}} does not admit any polynomial kernels with respect to either~$\abs{\rs}+\kappa$ or~$\abs{\bs}$, assuming {$\textsf{PH}\neq \Sigma_{\textsf{P}}^3$} (see, e.g.,~\cite{Cygan2015,DBLP:journals/talg/DomLS14}). Since our reductions are polynomial parameter transformations, they also imply the absence of polynomial kernels for many related problems. 

\begin{definition}[Polynomial Parameter Transformation]
Let~$P$ and~$Q$ be two parameterized problems. A polynomial parameter transformation from~$P$ to~$Q$ is an algorithm that, given an instance $(X,\kappa)$ of~$P$, outputs an instance $(X', \kappa')$ of~$Q$ such that:
\begin{itemize}
\item The algorithm runs in polynomial time.
\item $(X, \kappa)$ is a {\yesins} of~$P$ if and only if $(X', \kappa')$ is a {\yesins} of $Q$.
\item $\kappa'\leq g(\kappa)$ for some polynomial function~$g$ of~$\kappa$.
\end{itemize}
\end{definition}

\begin{lemma}[\cite{DBLP:conf/esa/BodlaenderTY09,Cygan2015,DBLP:journals/talg/DomLS14}]
Let~$P$ and~$Q$ be two parameterized problems such that the unparameterized version of~$P$ is {\memph\npc} and the unparameterized version of~$Q$ is in {\memph\np}. 
If there is a polynomial parameter transformation from~$P$ to~$Q$, then~$Q$ admits a polynomial kernel only if~$P$ does.
\end{lemma}

With these notions, and recalling that 
(1) the unparameterized version of {\prob{RBDS}} is {\npc}, and 
(2) {\prob{RBDS}} does not admit any polynomial kernels with respect to~$\abs{\rs}+\kappa$ or~$\abs{\bs}$ assuming {$\textsf{PH}\neq \Sigma_{\textsf{P}}^3$}, 
 our reductions imply the nonexistence of polynomial kernels for many problems, as summarized in Table~\ref{tab-kernelization-lower-bounds}.

\begin{table}
\caption{
Nonexistence of polynomial kernels for various voting problems for TSMR, assuming {$\textsf{PH} \neq \Sigma_{\textsf{P}}^3$}. 
Each listed parameter indicates that the corresponding problem does not admit a polynomial kernel with respect to that parameter. 
``No {\poly}-kernel'' stands for ``no polynomial kernel''. Theorems in brackets next to the parameters demonstrate that the reductions used constitute  polynomial parameter transformations, which imply the corresponding kernelization lower bounds.
}
\centering
{
\begin{tabular}{lllll} \toprule
&{\prob{CCAV}, \prob{DCAV}}
&{\prob{CCDV}, \prob{DCDV}}
&{\prob{CCAC}}
&\prob{Possible Winner}
\\ \midrule

No {\poly}-kernel &
$m+k$, $n_{\text{u}}$ (Thms.~\ref{thm-ccav-hard-last},~\ref{thm-dcav-hard-first})
&$m+k$, $n$ (Thms.~\ref{thm-ccdv-hard-deleted-first},~\ref{thm-dcdv-hard-deleted-first})
&$m_{\text{u}}+k$, $m_{\text{r}}$ (Thm.~\ref{thm-ccac-hard})

&$m$, $n$ (Thm.~\ref{thm-possible-hard}) \\ \bottomrule
\end{tabular}
}
\label{tab-kernelization-lower-bounds}
\end{table}

\subsection{Lower Bounds for Exact Algorithms}
Assuming the Strong Exponential Time Hypothesis (SETH), {\prob{RBDS}} cannot be solved in time $\bigos{(2-\epsilon)^{\abs{\bs}}}$ for any $\epsilon>0$~\cite{Cygan2015}. Moreover, unless~ETH fails, {\prob{RBDS}} cannot be solved in time $f(\kappa) \cdot \abs{\bs}^{\smallo{\kappa}}$ for any computable function~$f$ in~$\kappa$~\cite{DBLP:journals/jcss/ChenHKX06}. Combined with our reductions, these lower bounds suggest that brute-force algorithms for several election control problems are essentially optimal. Table~\ref{tab-algorithm-lower-bounds} summarizes the results.

\begin{table}[ht!]
\renewcommand{\tabcolsep}{2.5mm}
\caption{Lower bounds for exact algorithms for election control problems under TSMR. The theorems in brackets indicate the sources whose reductions imply the corresponding lower bounds.
}
\centering
{
\begin{tabular}{p{0.45\textwidth} p{0.45\textwidth}} \toprule
{\prob{CCAV}, \prob{DCAV}}
&{\prob{CCAC}}
\\ \midrule

no $\bigos{(2-\epsilon)^{n_{\text{u}}}}$-time algorithm for any $\epsilon >0$ unless SETH fails (Thms.~\ref{thm-ccav-hard-last},~\ref{thm-dcav-hard-first})
& no $\bigos{(2-\epsilon)^{m_{\text{u}}}}$-time algorithm  for any $\epsilon >0$ unless SETH fails (Thm.~\ref{thm-ccac-hard})
\\ \midrule

no $f(k)\cdot {n_{\text{u}}^{k}}$-time algorithm for any computable function $f$ unless ETH fails (Thms.~\ref{thm-ccav-hard-last},~\ref{thm-dcav-hard-first})
&no  $f(k)\cdot {m_{\text{u}}^{k}}$-time algorithm for any computable function $f$ unless ETH fails  (Thm.~\ref{thm-ccac-hard})
\\ 
\bottomrule
\end{tabular}
}
\label{tab-algorithm-lower-bounds}
\end{table}

\subsection{Lower Bounds for Approximation Algorithms}
\begin{table}[ht!]
\caption{Lower bounds on polynomial-time approximation algorithms for election control problems for TSMR. Theorems in brackets indicate sources whose reductions imply the corresponding lower bounds.
}
\begin{center}
{
\begin{tabular}{p{0.3\textwidth}p{0.3\textwidth}p{0.3\textwidth}} \toprule
{\prob{CCAV}, \prob{DCAV}}
&{\prob{CCDV}, \prob{DCDV}}
&{\prob{CCAC}}
\\ \midrule

no $(1-\epsilon)\ln m$-approximation for any $\epsilon>0$ unless $\poly= \np$ (Thms.~\ref{thm-ccav-hard-last},~\ref{thm-dcav-hard-first})
& no $(1-\epsilon)\ln m$-approximation for any $\epsilon>0$ unless $\poly= \np$   (Thms.~\ref{thm-ccdv-hard-deleted-first},~\ref{thm-dcdv-hard-deleted-first})
& no $(1-\epsilon)\ln m_{\text{r}}$-approximation for any $\epsilon>0$ unless $\poly= \np$  (Thm.~\ref{thm-ccac-hard})
\\ \midrule
\end{tabular}
}
\end{center}
\label{tab-poly-inapp}
\end{table}

For each control problem, we consider its optimal version, in which the goal is to add or delete the minimum number of votes or candidates to make the distinguished candidate a winner (constructive) or prevent them from winning (destructive).

Let {\prob{Optimal RBDS}} be the optimal version of {\prob{RBDS}}, where the objective is to select a minimum number of blue vertices that dominate all red vertices. 
It is well known that, for any $\epsilon>0$, unless $\np\subseteq \dtime(n^{\bigo{\log\log n}})$, {\prob{Optimal RBDS}} cannot be approximated in polynomial time within a factor of $(1-\epsilon)\ln \abs{\rs}$~\cite{DBLP:journals/jacm/Feige98}. This result was recently strengthened by Dinur and Steurer~\cite{DBLP:conf/stoc/DinurS14}, who showed that the same lower bound holds even under the weaker assumption that $\poly \neq \np$. 

Building on these results and our earlier reductions, we derive several inapproximability results for the optimal versions of the control problems, as summarized in Table~\ref{tab-poly-inapp}.

\section{Conclusion}
\label{sec-con}
We have investigated the (parameterized) complexity of several natural voting problems under the recently proposed TSMR rule, considering a variety of problem-specific parameters. Our analysis provides a comprehensive map of the complexity landscape, including polynomial-time algorithms, {\nphns} results, and {\wbhns} outcomes. Notably, our hardness results often hold even when the distinguished candidate is positioned at the very beginning or the very end of the agenda, offering a thorough characterization of the computational challenges associated with strategic behavior under TSMR (see Table~\ref{tab-results}).

These findings allow for a meaningful comparison between TSMR and the two well-studied sequential rules, namely the amendment rule and the successive rule. In particular, TSMR resists most control problems, while remaining vulnerable to  {\prob{Agenda Control}} and {\prob{Coalition Manipulation}}. Moreover,  {\prob{Possible Winner-TSMR}} is {\nph}, whereas {\prob{Necessary Winner-TSMR}} is solvable in polynomial time. Taken together, these results suggest that TSMR performs comparably to the amendment rule and the successive rule in terms of resilience to strategic behavior and the computational complexity of determining winners under uncertainty.

Despite its promising theoretical properties, it is important to recognize that applying TSMR in practical settings of sequential decision-making---such as parliamentary procedures---may present logistical challenges. One key reason behind the successful deployment of the amendment rule and the successive rule in real-world parliamentary decision-making is their simplicity. Both rules can be implemented in a linear number of iterations, with one or more candidates potentially eliminated in each round without requiring further consideration. This structure reduces the cognitive burden on voters. Moreover, in each iteration, a voter needs only to compare the current candidate with some of its successors. In particular, the amendment rule requires each voter to compare only two candidates per iteration. Under the successive rule, a voter rejects the current candidate if and only if they prefer any of the remaining candidates in the agenda over the one under consideration. 
In contrast, applying TSMR imposes a significantly greater cognitive load. It remains unclear whether TSMR can be implemented with the same linear efficiency as the amendment and successive rules, where each iteration involves only simple, localized comparisons.

Nevertheless, rather than limiting our perspective, we should remain open to the potential of newly proposed voting rules, as they may prove valuable in a wide range of practical settings---both now and in the future---even beyond their originally intended scope. In particular, as artificial intelligence continues to evolve, new domains are emerging where robust and expressive voting rules like TSMR could play a significant role.

A key limitation of our current work is its theoretical focus. Future studies are needed to evaluate the practical performance and applicability of TSMR in real-world scenarios, possibly through empirical studies, simulations, or user-centered experiments.

As for other promising directions for future research, one could investigate how domain restrictions on voters' preferences influence the computational complexity of strategic problems under TSMR. For an overview of commonly studied domain restrictions and their implications, we refer the reader to the survey by Karpov~\cite{DBLP:journals/aarc/Karpov22}, and the survey by Elkind, Lackner, and Peters~\cite{DBLP:journals/corr/abs-2205-09092}. Additionally, resolving the open questions discussed in Sections~\ref{sec-possible} and~\ref{sec-remarks-incompleteness} remains an important avenue for future work.


\section*{Acknowledgments}
The author sincerely thanks the anonymous reviewers of AAMAS 2023 for their valuable feedback on an earlier version of this paper. 
The author also extends heartfelt gratitude to the anonymous reviewers of Autonomous Agents and Multi-Agent Systems for their insightful and constructive comments, which significantly enhanced the quality of this work.

%
%




\section*{Appendix: Proofs Omitted from the Main Body}

\noindent{\bf{Theorem}~\ref{thm-ccdv-hard-first}.}
{\it{{\prob{CCDV-TSMR}} is {\memph\wbh} with respect to the number of votes not deleted, even when the distinguished candidate is first in the agenda.}}
\smallskip

\begin{proof}
We prove the theorem via a reduction from {\prob{RBDS}}. Let $(G, \kappa)$ be an instance of {\prob{RBDS}}, where~$G$ is a bipartite graph with the bipartition~$(R, B)$.  As in the proof of Theorem~\ref{thm-ccdv-hard-deleted-first}, we assume that every red vertex has degree exactly~$\ell$ for some positive integer~$\ell$. 
We construct an instance of {\prob{CCDV-TSMR}} as follows. The candidate set is $C=\rs\cup \{p, q\}$, and the agenda is $\rhd=(p,\overrightarrow{\rs},q)$. We create three groups of votes:
\begin{itemize}
    \item  a multiset~$V_1$ of~$\kappa$ votes with the preference $p\Succ q\Succ \overleftarrow{\rs}$;
    \item a singleton~$V_2$ of one vote with the preference $\overleftarrow{\rs}\Succ p\Succ q$; and
    \item  for every blue vertex $b\in \bs$, one vote $\succ_b$ with the preference
    $q\Succ \left(\overleftarrow{\rs}\setminus \neighbor{G}{b}\right) \Succ p \Succ \left(\overleftarrow{\rs}[\neighbor{G}{b}]\right)$.
\end{itemize}
Let~$V$ denote the multiset of the above $\abs{\bs}+\kappa+1$ votes. For a given $\bs'\subseteq \bs$, we use $V(\bs')=\{\succ_b\; \setmid b\in \bs'\}$ to denote the multiset of votes corresponding to~$\bs'$.
We complete the construction by setting $k=\abs{\bs}-\kappa$. The instance of {\prob{CCDV-TSMR}} is $((C, V), p, \rhd, k)$, which can be constructed in polynomial time.
It remains to show the correctness of the reduction.

$(\Rightarrow)$ Assume that there exists $\bs'\subseteq \bs$ such that $\abs{\bs'}=\kappa$ and~$\bs'$ dominates~$\rs$ in~$G$. Let $\elec=(C, V_1\cup V_2\cup V(\bs'))$ be the election obtained from $(C, V)$ by deleting the $k=\abs{\bs}-\kappa$ votes corresponding to the blue vertices in $\bs\setminus \bs'$. We show that~$p$ is the TSMR winner of~$\elec$ with respect to the agenda~$\rhd$. It suffices to show that~$p$ beats every other candidate in~$\elec$. Let $r\in \rs$. Since~$\bs'$ dominates~$\rs$, there exists some $b\in \bs'$ that dominates~$r$; hence~$\succ_b$ ranks~$p$ before~$r$. Therefore, at least $\abs{V_1}+1=\kappa+1$ votes in~$\elec$ rank~$p$ before~$r$. Moreover, $\abs{V_1}+\abs{V_2}=\kappa+1$ votes in~$\elec$ rank~$p$ before~$q$. Since $\abs{V_1\cup V_2\cup V(\bs')}=2\kappa+1$, candidate~$p$ beats every other candidate in~$\elec$. Therefore,~$p$ is the TSMR winner of~$\elec$.

$(\Leftarrow)$ Assume that there exists $V'\subseteq V$ such that $\abs{V'}\leq k=\abs{\bs}-\kappa$ and~$p$ is the TSMR winner of the election $\elec=(C, V\setminus V')$. Observe first that  $V'\subseteq V(\bs)$ and $\abs{V'}=k$; otherwise~$q$ is not beaten by any of her predecessors and thus~$q$ wins~$\elec$, a contradiction. Let $\bs'\subseteq \bs$ be such that $\abs{\bs'}=k=\abs{\bs}-\kappa$ and $V(\bs')=V'$. Let $\overline{\bs}=B\setminus \bs'$. Obviously, $\abs{\overline{\bs}}=\kappa$ and $\abs{V\setminus V'}=2\kappa+1$. By construction, regardless of which~$k$ votes are contained in~$V(\bs')$, every candidate from $C\setminus \{p\}$ beats all her predecessors in $C\setminus \{p\}$. Since~$p$ appears first in the agenda and wins~$\elec$, it follows that~$p$ beats every other candidate in~$\elec$. We claim that~$\overline{\bs}$ dominates~$\rs$. Suppose, for contradiction, that there exists a red vertex $r\in \rs$ not dominated by any vertex from~$\overline{\bs}$. Then, by construction, all votes in~$V(\overline{\bs})$ rank~$r$ before~$p$. As a consequence, there are $\abs{\overline{\bs}}+\abs{V_2}=\kappa+1$ votes ranking~$r$ before~$p$ in~$\elec$, contradicting that~$p$ beats~$r$ in~$\elec$.
\end{proof}

\bigskip

\noindent{\bf{Theorem}~\ref{thm-ccdv-hard}.} 
{\it{{\prob{CCDV-TSMR}} is {\emph{\wbh}} with respect to the number of deleted votes, even when the distinguished candidate is last in the agenda.}}
\smallskip

\begin{proof}
We prove the theorem by a reduction from {\prob{RBDS}}. 
Let $I=(G, \kappa)$ be an instance of {\prob{RBDS}}, where~$G$ is a bipartite graph with the bipartition~$(\rs, \bs)$. 
As in the proof of Theorem~\ref{thm-ccdv-hard-deleted-first}, we assume that~$G$ contains no isolated vertices, $\kappa \geq 4$, and every red vertex has degree~$\ell \geq 1$. 
Let $C = \rs \cup \{p, q\}$, and let~$\rhd$ be an agenda on~$C$ in which~$p$ appears last; the relative order of the remaining candidates is immaterial. 
We create a multiset~$V$ of $2\abs{\bs}+2\ell+\kappa$ votes:
\begin{itemize}
\item $\abs{\bs}+1$ votes with the preference $\overrightarrow{\rs}\Succ p\Succ q$;
\item $\ell+\kappa$ votes with the preference $q\Succ p\Succ \overrightarrow{\rs}$;
\item $\ell-1$ votes with the preference $p\Succ q\Succ \overrightarrow{\rs}$; and 
\item for each blue vertex $b\in \bs$, one vote~$\succ_b$ with the preference
$q\Succ \left(\overrightarrow{\rs}[\neighbor{G}{b}]\right) \Succ p \Succ \left(\overrightarrow{\rs}\setminus \neighbor{G}{b}\right)$.
\end{itemize}
For a given $\bs' \subseteq \bs$, let $V(\bs') = \{\succ_b\, \setmid b \in \bs'\}$. Since every red vertex has degree~$\ell$, for each $r \in \rs$ exactly $\abs{\bs} - \ell$ votes in~$V(\bs)$ rank~$p$ before~$r$.

Let $k=\kappa$. The instance of {\prob{CCDV-TSMR}} is $g(I)=((C,V), p, \rhd, k)$. 
We prove the correctness of the reduction.

$(\Rightarrow)$ Assume that there exists $\bs'\subseteq \bs$ of cardinality~$\kappa$ that dominates~$\rs$ in~$G$. Let $\elec=(C, V\setminus V(\bs'))$. Clearly, $\abs{V\setminus V(\bs')}=2\abs{\bs}+2\ell$. To show that $g(I)$ is a {\yesins}, it suffices to show that~$p$ is not beaten by any other candidate in~$\elec$. 

Since all votes in~$V(\bs')$ rank~$q$ before~$p$, the set $V\setminus V(\bs')$ contains $(\ell+\kappa)+(\abs{\bs}-\kappa)=\abs{\bs}+\ell$ votes ranking~$q$ before~$p$, and thus~$p$ ties with~$q$ in~$\elec$. Next, consider candidates in~$\rs$. Let $r\in \rs$. Since~$\bs'$ dominates~$\rs$, there exists a vertex $b\in \bs'$ that dominates~$r$. By construction,~$r$ is ranked before~$p$ in the vote $\succ_b\, \in V(\bs')$. Hence, at most $\abs{V(\bs')} - 1 = \kappa - 1$ votes in~$V(\bs')$ rank~$p$ before~$r$. 
Consequently, the number of votes in~$V\setminus V(\bs')$ ranking~$p$ before~$r$ is at least 
$(\ell+\kappa)+(\ell-1)+(\abs{\bs}-\ell)-(\kappa-1)=\abs{\bs}+\ell$, implying that~$p$ beats or ties with~$r$ in~$\elec$. 

$(\Leftarrow)$ Assume there exists $V'\subseteq V$ such that $\abs{V'}\leq k$ and~$p$ is the TSMR winner of $\elec=(C, V\setminus V')$ with respect to~$\rhd$. Since~$p$ is last in~$\rhd$, it follows that~$p$ beats or ties with every other candidate in~$\elec$. Consequently, all votes in~$V'$ rank~$q$ before~$p$, and necessarily $\abs{V'}=k=\kappa$; otherwise~$p$ would be beaten by~$q$ in~$\elec$. 
There are two groups of votes ranking~$q$ before~$p$: those corresponding to the blue vertices, and those with preference $q\Succ p\Succ \overrightarrow{\rs}$. We may assume that all votes in~$V'$ are from~$V(\bs)$. Indeed, if~$V'$ contained some vote with preference $q\Succ p\Succ \overrightarrow{\rs}$, we can obtain another feasible solution by replacing this vote with any vote in~$V(\bs)\setminus V'$. 
 
Let $r\in \rs$. Since there are $(\abs{\bs}+1)+\ell$ votes in~$V$ ranking~$r$ before~$p$, and $\abs{V\setminus V'}=2\abs{\bs}+2\ell$, there is at least one vote $\succ_b\,\in V'$ that ranks~$r$ before~$p$. By construction, the vertex~$b$ corresponding to~$\succ_b$ dominates~$r$. Hence, $\{b\in \bs \setmid\,  \succ_b\in V'\}$ dominates~$\rs$, implying that~$I$ is a {\yesins} of {\prob{RBDS}}.
\end{proof}
\bigskip

\noindent{\bf{Theorem}~\ref{thm-ccdv-wbh-not-deleted-last}}.
{\it{{\prob{CCDV-TSMR}} is {\memph\wbh} with respect to the number of votes not deleted, even when the distinguished candidate is last in the agenda.}}
\smallskip

\begin{proof}
We prove the theorem by a reduction from {\prob{RBDS}}. Let $(G, \kappa)$ be an instance of {\prob{RBDS}}, where~$G$ is a bipartite graph with the bipartition $(\rs, \bs)$. Without loss of generality, assume that $\abs{\bs}\geq \kappa\geq 2$. 
We create an instance of {\prob{CCDV-TSMR}} as follows. 
Let $C=\rs\cup \{q, p\}$. Let~$\rhd$ be an agenda on~$C$ in which~$p$ appears last. 
We create the following votes:
\begin{itemize}
\item a multiset~$V_1$ of $\kappa-1$ votes with the preference $p\Succ q\Succ \overrightarrow{\rs}$;
\item a singleton~$V_2$ of one vote with the preference $\overrightarrow{\rs}\Succ p\Succ q$; and
\item for each blue vertex $b\in \bs$, one vote $\succ_b$ with the preference
$q\Succ \left(\overrightarrow{\rs}\setminus \neighbor{G}{b}\right) \Succ p\Succ \left(\overrightarrow{\rs}[\neighbor{G}{b}]\right)$.
\end{itemize}
Since $\kappa\geq 2$,~$V_1$ is nonempty. For a given $\bs'\subseteq \bs$, let $V(\bs')=\{\succ_b\, \setmid b\in \bs'\}$. Let $V=V_1\cup V_2\cup V(\bs)$. Clearly, $\abs{V}=\abs{\bs}+\kappa$. Finally, let $k=\abs{\bs}-\kappa$. The instance of {\prob{CCDV-TSMR}} is $((C, V), p, \rhd, k)$. We prove the correctness as follows.

$(\Rightarrow)$ Assume that there exists $\bs'\subseteq \bs$ such that~$\abs{\bs'}=\kappa$ and~$\bs'$ dominates~$\rs$ in~$G$. Let $V'=V_1\cup  V_2\cup V(\bs')$, and let $\elec=(C, V')$. Clearly, $\abs{V'}=2\kappa$. We claim that~$p$ is the TSMR winner of~$\elec$ with respect to the agenda~$\rhd$. As~$p$ appears last in~$\rhd$, it suffices to show that~$p$ is unbeaten by any other candidate in~$\elec$. It is clear that~$p$ ties with~$q$ in~$\elec$. Let $r\in \rs$ be a red vertex. As~$\bs'$ dominates~$\rs$, there exists $b\in \bs'$ dominating~$r$. By construction,~$p$ is ranked before~$r$ in the vote~$\succ_b$. Therefore, there are at least $\abs{V_1}+1=\kappa$ votes ranking~$p$ before~$r$ in~$V'$, implying that~$p$ is not beaten by~$r$. As this holds for all $r\in \rs$, the forward direction follows.

$(\Leftarrow)$ Assume that there exists $V'\subseteq V$ such that $\abs{V'}\geq 2\kappa$ and~$p$ is the TSMR winner of $(C, V')$ with respect to~$\rhd$. Since $\abs{V_1}+\abs{V_2}=\kappa$ and all votes in~$V(\bs)$ rank~$q$ first, it follows that $(V_1\cup V_2)\subseteq V'$. Moreover,~$V'$ contains exactly~$\kappa$ votes from~$V(\bs)$; otherwise~$q$ would be the winner of $(C,V')$, contradicting that $p$ wins. Let $V(\bs')=V'\cap V(\bs)$, where $\bs'\subseteq \bs$. We have $\abs{V({\bs')}}=\kappa$. We claim that~$\bs'$ dominates~$\rs$ in~$G$. Suppose, for contradiction, that there exists $r\in \rs$ not dominated by any vertex in~$\bs'$. By construction,~$r$ is ranked before~$p$ in all the~$\kappa$ votes of~$V({\bs'})$. Together with the vote in~$V_2$, there are $\kappa+1$ votes in~$V'$ ranking~$r$ before~$p$, implying that~$r$ beats~$p$. However, in this case,~$p$ cannot be the winner of $(C, V')$, a contradiction. Since $\abs{\bs'}=\kappa$, $(G, \kappa)$ is a {\yesins} of {\prob{RBDS}}.
\end{proof}

\noindent{\bf{Theorem}~\ref{thm-dcdv-hard-deleted-first}.} 
{\it{\prob{DCDV-TSMR} is {\memph{\wbh}} with respect to the number of deleted votes, as long as the distinguished candidate is not last in the agenda.}}
\medskip

\begin{proof}
The reduction is identical to the one used in the proof of Theorem~\ref{thm-ccdv-hard}, except that~$q$ is now the distinguished candidate and we assume $\ell\ge 2$, which does not affect the {\wbhns} of {\prob{RBDS}}. The correctness of the reduction relies on the fact that, regardless of which at most~$k$ votes are deleted,~$q$ beats all candidates in~$\rs$, assuming $\ell\geq 2$. This leaves~$p$ as the sole candidate preventing~$q$ from winning. Moreover, this remains true as long as~$q$ is not last in the agenda. 
\end{proof}


\bigskip

\noindent{{\bf{Theorem}}~\ref{thm-dcdv-wbh-not-deleted-first}}. 
{\it{{\prob{DCDV-TSMR}} is {\memph\wbh} with respect to the number of votes not deleted, as long as the distinguished candidate is not last in the agenda.}}
\smallskip

\begin{proof}
The reduction is identical to the one used in the proof of Theorem~\ref{thm-ccdv-wbh-not-deleted-last}, except that $q$ is the distinguished candidate. 
\end{proof}

\bigskip

\noindent{{\bf{Theorem}}~\ref{thm-possible-hard-second-last}}. 
{\it{{\prob{Possible Winner-TSMR}} is {\memph{\nph}}, even for complete agendas with the distinguished candidate in the penultimate position.}}
\smallskip

\begin{proof}
We prove the theorem via a reduction from {\prob{RBDS}}. Let $I=(G, \kappa)$ be an instance of {\prob{RBDS}}, where~$G$ is a bipartite graph with the bipartition~$(\rs, \bs)$ and $1\leq \kappa\leq \abs{\bs}$. As in the proof of Theorem~\ref{thm-ccdv-hard-deleted-first}, we assume that every red vertex in~$R$ has degree $\ell\geq 1$.  
We construct an instance of {\prob{Possible Winner-TSMR}} as follows. Let $C=\rs\cup \{p, q, q'\}$ and let
$\rhd=(q', \overrightarrow{\rs}, p, q)$. We create five groups of votes, where only the first group contains partial votes:
\begin{itemize}
    \item for every $b\in \bs$, one partial vote~$\succ_b$ with the preference
    \[\left(\overrightarrow{\rs}\setminus \neighbor{G}{b}\right) \Succ q'\]
    and
    \[q\Succ p\Succ \left(\overrightarrow{\rs}[\neighbor{G}{b}]\right);\]
    \item a multiset~$V_1$ of $\abs{\bs}+1$ votes with the preference $q'\Succ q\Succ \overrightarrow{\rs}\Succ p$;
    \item a multiset~$V_2$ of~$2\kappa$ votes with the preference $q\Succ p \Succ \overrightarrow{\rs} \Succ q'$;
    \item a multiset~$V_3$ of $\kappa$ votes with the preference $q\Succ p\Succ q'\Succ \overrightarrow{\rs}$;
    \item a multiset~$V_4$ of $\kappa$ votes with the preference $\overrightarrow{\rs}\Succ p\Succ q'\Succ q$.
\end{itemize}
For $\bs'\subseteq \bs$, let~$V(\bs')=\{\succ_b \setmid b\in \bs'\}$. Let~$V$ be the multiset of all the constructed $2\abs{\bs} + 4\kappa + 1$ votes, and define $V(\overline{\bs}) = V \setminus V(\bs)$. The resulting instance is $g(I) = ((C, V), p, \rhd)$, which is computable in polynomial time. 
We prove that~$I$ is a {\yesins} of {\prob{RBDS}} if and only if~$g(I)$ is a {\yesins} of {\prob{Possible Winner-TSMR}}.

$(\Rightarrow)$ Suppose that~$I$ is a {\yesins}. That is, there is a subset $\bs'\subseteq \bs$ such that $\abs{\bs'}=\kappa$, and~$\bs'$ dominates~$\rs$ in~$G$. 
We complete each partial vote $\succ_b\, \in V(\bs)$, where $b \in \bs$, as follows:
\begin{itemize}
    \item if $b\in \bs'$, complete it as $\left(\overrightarrow{\rs}\setminus \neighbor{G}{b}\right) \Succ q' \Succ q \Succ p \Succ \left(\overrightarrow{\rs}[\neighbor{G}{b}]\right)$
    \item otherwise, complete it as $q \Succ p \Succ \left(\overrightarrow{\rs}\setminus \neighbor{G}{b}\right) \Succ q' \Succ \left(\overrightarrow{\rs}[\neighbor{G}{b}]\right)$.
\end{itemize}
Under this completion,~$p$ beats all its predecessors in~$\rhd$, while~$q$ is beaten by its predecessor~$q'$. Hence,~$p$ is the TSMR winner with respect to~$\rhd$, and thus~$g(I)$ is a {\yesins}.

$(\Leftarrow)$ Assume that $g(I)$ is a {\yesins}. That is, there is a completion~$V'$ of~$V(\bs)$ such that~$p$ is the TSMR winner of the election $\elec=(C, V(\overline{\bs})\cup V')$, with respect to~$\rhd$. Regardless of the completions,~$q$ beats all its predecessors except~$q'$. Since~$p$ wins~$\elec$, and~$q$ appears last in~$\rhd$, it must be that~$q'$ beats~$q$ in~$\elec$. Hence, at least~$\kappa$ partial votes in~$V(\bs)$ are completed with~$q'$ ranked before~$q$. Each partial vote $\succ_b\, \in V(\bs)$ admits only one such completion, namely 
\begin{equation}
\label{eq-a}
    \left(\overrightarrow{\rs}\setminus \neighbor{G}{b}\right) \Succ q' \Succ q \Succ p \Succ \left(\overrightarrow{\rs}[\neighbor{G}{b}]\right).
\end{equation}
Let $\bs' \subseteq \bs$ be the set of vertices whose corresponding partial votes are completed this way, Write $\abs{\bs'} = \kappa + t$ for some $t\geq 0$. Since~$p$ wins~$\elec$ and~$\abs{V}$ is odd,~$p$ must beat every candidate in~$\rs$. For each $r\in \rs$, exactly~$3\kappa$ votes in $V(\overline{\bs})$ (namely those in~$V_2 \cup V_3$) rank~$p$ before~$r$. To ensure that~$p$ beats~$r$ in~$\elec$, at least $\abs{\bs}-\kappa+1$ partial vote completions must rank $p$ before~$r$. It follows that at least $t + 1$ completions of partial votes in~$V(\bs')$ rank~$p$ before~$r$. By~\eqref{eq-a},~$p$ is ranked before~$r$ in a completion of some $\succ_b$ where $b\in \bs'$ if and only if~$r$ is a neighbor of~$b$ in~$G$. Hence, every~$r\in \rs$ has at least $t+1$ neighbors from~$\bs'$ in~$G$. 
By removing any $t$ vertices from $\bs'$, we obtain a subset of size $\kappa$ that dominates $\rs$.  Thus~$I$ is a {\yesins} of {\prob{RBDS}}.
\end{proof}


\begin{thebibliography}{99}\itemsep=-1pt

\bibitem{DBLP:journals/theoretics/AlmanW24}
J.~Alman and V.~{Vassilevska Williams}.
\newblock A refined laser method and faster matrix multiplication.
\newblock {\em TheoretiCS}, 3:Nr.~21, 2024.

\bibitem{DBLP:journals/ai/AzizGMMSW18}
H.~Aziz, S.~Gaspers, S.~Mackenzie, N.~Mattei, P.~Stursberg, and T.~Walsh.
\newblock Fixing balanced knockout and double elimination tournaments.
\newblock {\em Artif. Intell.}, 262:1--14, 2018.

\bibitem{DBLP:books/sp/BG2018}
J.~Bang{-}Jensen and G.~Z. Gutin, editors.
\newblock {\em Classes of Directed Graphs}.
\newblock Springer Monographs in Mathematics. Springer,
  2018.

\bibitem{Bartholdi92howhard}
J.~J. {Bartholdi III}, C.~A. Tovey, and M.~A. Trick.
\newblock How hard is it to control an election?
\newblock {\em Math. Comput. Model.}, 16(8-9):27--40, 1992.

\bibitem{BARTHOLDI89}
J.J. {Bartholdi III}, C.A. Tovey, and M.A. Trick.
\newblock The computational difficulty of manipulating an election.
\newblock {\em Soc. Choice Welfare}, 6(3):227--241, 1989.

\bibitem{BaumeisterR2016}
D.~Baumeister and J.~Rothe.
\newblock Preference aggregation by voting.
\newblock In J.~Rothe, editor, {\em Economics and Computation: An Introduction
  to Algorithmic Game Theory, Computational Social Choice, and Fair Division},
  chapter~4, pages 197--325. Springer, 2016.

\bibitem{DBLP:conf/ijcai/BetzlerHN09}
N.~Betzler, S.~Hemmann, and R.~Niedermeier.
\newblock A multivariate complexity analysis of determining possible winners
  given incomplete votes.
\newblock In {\em IJCAI}, pages 53--58, 2009.

\bibitem{Black1958}
D.~Black.
\newblock {\em The Theory of Committees and Elections}.
\newblock Cambridge University Press, London, 1958.

\bibitem{DBLP:conf/esa/BodlaenderTY09}
Hans~L. Bodlaender, St{\'{e}}phan Thomass{\'{e}}, and Anders Yeo.
\newblock Kernel bounds for disjoint cycles and disjoint paths.
\newblock In {\em ESA}, pages 635--646, 2009.

\bibitem{DBLP:conf/ijcai/BoehmerBFN21}
N.~Boehmer, R.~Bredereck, P.~Faliszewski, and R.~Niedermeier.
\newblock Winner robustness via swap- and shift-bribery: Parameterized counting
  complexity and experiments.
\newblock In {\em IJCAI}, pages 52--58, 2021.

\bibitem{DBLP:journals/jair/BredereckCNW17}
R.~Bredereck, J.~Chen, R.~Niedermeier, and T.~Walsh.
\newblock Parliamentary voting procedures: Agenda control, manipulation, and
  uncertainty.
\newblock {\em J. Artif. Intell. Res.}, 59:133--173, 2017.

\bibitem{DBLP:conf/atal/CarletonCHNTW23a}
Benjamin Carleton, Michael~C. Chavrimootoo, Lane~A. Hemaspaandra, David~E.
  Narv{\'{a}}ez, Conor Taliancich, and Henry~B. Welles.
\newblock Search versus search for collapsing electoral control types.
\newblock In {\em AAMAS}, pages 2682--2684, 2023.

\bibitem{DBLP:journals/jair/CarletonCHNTW24}
Benjamin Carleton, Michael~C. Chavrimootoo, Lane~A. Hemaspaandra, David~E.
  Narv{\'{a}}ez, Conor Taliancich, and Henry~B. Welles.
\newblock Separating and collapsing electoral control types.
\newblock {\em J. Artif. Intell. Res.}, 81:71--116, 2024.

\bibitem{DBLP:journals/tdasci/ChakrabortyDKKR21}
V.~Chakraborty, T.~Delemazure, B.~Kimelfeld, P.~G. Kolaitis, K.~Relia, and
  J.~Stoyanovich.
\newblock Algorithmic techniques for necessary and possible winners.
\newblock {\em Trans. Data Sci.}, 2(3):Nr.~22, 2021.

\bibitem{DBLP:conf/atal/ChakrabortyK21}
V.~Chakraborty and P.~G. Kolaitis.
\newblock Classifying the complexity of the possible winner problem on partial
  chains.
\newblock In {\em AAMAS}, pages 297--305, 2021.

\bibitem{DBLP:journals/jcss/ChenHKX06}
J.~Chen, X.~Huang, I.~A. Kanj, and G.~Xia.
\newblock Strong computational lower bounds via parameterized complexity.
\newblock {\em J. Comput. Syst. Sci.}, 72(8):1346--1367, 2006.

\bibitem{DBLP:conf/ijcai/ConitzerLX09}
V.~Conitzer, J.~Lang, and L.~Xia.
\newblock How hard is it to control sequential elections via the agenda?
\newblock In {\em IJCAI}, pages 103--108, 2009.

\bibitem{DBLP:journals/jacm/ConitzerSL07}
V.~Conitzer, T.~Sandholm, and J.~Lang.
\newblock When are elections with few candidates hard to manipulate?
\newblock {\em J.~ACM.}, 54(3):Nr.~14, 2007.

\bibitem{Conitzer2016}
V.~Conitzer and T.~Walsh.
\newblock Barriers to manipulation in voting.
\newblock In F.~Brandt, V.~Conitzer, U.~Endriss, J.~Lang, and A.~Procaccia,
  editors, {\em Handbook of Computational Social Choice}, chapter~6, pages
  127--145. Cambridge University Press, 2016.

\bibitem{Cygan2015}
M.~Cygan, F.~V. Fomin, L.~Kowalik, D.~Lokshtanov, D.~Marx, M.~Pilipczuk,
  M.~Pilipczuk, and S.~Saurabh.
\newblock {\em Lower Bounds Based on the Exponential-Time Hypothesis},
  chapter~14, pages 467--521.
\newblock Springer, 2015.

\bibitem{DBLP:books/sp/CyganFKLMPPS15}
M.~Cygan, F.~V. Fomin, L.~Kowalik, D.~Lokshtanov, D.~Marx, M.~Pilipczuk,
  M.~Pilipczuk, and S.~Saurabh.
\newblock {\em Parameterized Algorithms}.
\newblock Springer, 2015.

\bibitem{DBLP:conf/stoc/DinurS14}
I.~Dinur and D.~Steurer.
\newblock Analytical approach to parallel repetition.
\newblock In {\em STOC}, pages 624--633, 2014.

\bibitem{DBLP:journals/talg/DomLS14}
M.~Dom, D.~Lokshtanov, and S.~Saurabh.
\newblock Kernelization lower bounds through colors and {IDs}.
\newblock {\em {ACM} Trans. Algorithms}, 11(2):Nr.~13, 2014.

\bibitem{DBLP:conf/lata/Downey12}
R.~Downey.
\newblock A parameterized complexity tutorial.
\newblock In {\em LATA}, pages 38--56, 2012.

\bibitem{fellows2001}
R.~G. Downey, M.~R. Fellows, and U.~Stege.
\newblock Parameterized complexity: A framework for systematically confronting
  computational intractability.
\newblock In {\em Contemporary Trends in Discrete Mathematics: From DIMACS and
  DIMATIA to the Future}, pages 49--99, 1999.

\bibitem{DBLP:journals/ai/ElkindGORV21}
E.~Elkind, J.~Gan, S.~Obraztsova, Z.~Rabinovich, and A.~A. Voudouris.
\newblock Protecting elections by recounting ballots.
\newblock {\em Artif. Intell.}, 290:103401, 2021.

\bibitem{DBLP:journals/corr/abs-2205-09092}
E.~Elkind, M.~Lackner, and D.~Peters.
\newblock Preference restrictions in computational social choice: {A} survey.
\newblock {\em CoRR}, 2022.

\bibitem{DBLP:journals/aamas/ErdelyiNRRYZ21}
G.~Erd{\'{e}}lyi, M.~Neveling, C.~Reger, J.~Rothe, Y.~Yang, and R.~Zorn.
\newblock Towards completing the puzzle: Complexity of control by replacing,
  adding, and deleting candidates or voters.
\newblock {\em Auton. Agent Multi-Ag.}, 35(2):Nr.~41, 2021.

\bibitem{handbookofcomsoc2016Cha7FR}
P.~Faliszewski and J.~Rothe.
\newblock Control and bribery in voting.
\newblock In F.~Brandt, V.~Conitzer, U.~Endriss, J.~Lang, and A.~Procaccia,
  editors, {\em Handbook of Computational Social Choice}, chapter~7, pages
  146--168. Cambridge University Press,  2016.

\bibitem{Farquharson1969}
R.~Farquharson.
\newblock {\em Theory of Voting}.
\newblock Yale University Press, 1969.

\bibitem{DBLP:journals/jacm/Feige98}
Uriel Feige.
\newblock A threshold of ln \emph{n} for approximating set cover.
\newblock {\em J. {ACM}}, 45(4):634--652, 1998.

\bibitem{DBLP:conf/ijcai/FitzsimmonsH22}
Z.~Fitzsimmons and E.~Hemaspaandra.
\newblock Insight into voting problem complexity using randomized classes.
\newblock In {\em IJCAI}, pages 293--299, 2022.

\bibitem{garey}
M.~Garey and D.~Johnson.
\newblock {\em Computers and Intractability: A Guide to the Theory of
  {NP}-Completeness}.
\newblock W. H. Freeman, New York, 1979.

\bibitem{DBLP:journals/ai/HemaspaandraHR07}
E.~Hemaspaandra, L.~A. Hemaspaandra, and J.~Rothe.
\newblock Anyone but him: The complexity of precluding an alternative.
\newblock {\em Artif. Intell.}, 171(5-6):255--285, 2007.

\bibitem{DBLP:journals/jcss/HemaspaandraHR22}
E.~Hemaspaandra, L.~A. Hemaspaandra, and J.~Rothe.
\newblock The complexity of online bribery in sequential elections.
\newblock {\em J. Comput. Syst. Sci.}, 127:66--90, 2022.

\bibitem{HoranSprumont2022}
S.~Horan and Y.~Sprumont.
\newblock Two-stage majoritarian chocie.
\newblock {\em Theoretical Economics}, 17(2):521--537, 2022.

\bibitem{DBLP:journals/aarc/Karpov22}
A.~V. Karpov.
\newblock Structured preferences: {A} literature survey.
\newblock {\em Autom. Remote. Control.}, 83(9):1329--1354, 2022.

\bibitem{Konczak05votingprocedures}
K.~Konczak and J.~Lang.
\newblock Voting procedures with incomplete preferences.
\newblock In {\em M-PREF}, 2005.

\bibitem{DBLP:journals/tcs/LiuFZL09}
H.~Liu, H.~Feng, D.~Zhu, and J.~Luan.
\newblock Parameterized computational complexity of control problems in voting
  systems.
\newblock {\em Theor. Comput. Sci.}, 410(27-29):2746--2753, 2009.

\bibitem{DBLP:journals/dam/ManurangsiS23}
P.~Manurangsi and W.~Suksompong.
\newblock Fixing knockout tournaments with seeds.
\newblock {\em Discret. Appl. Math.}, 339:21--35, 2023.

\bibitem{DBLP:journals/ipl/Matousek91}
J.~Matou\v{s}ek.
\newblock Computing dominances in ${E}^n$.
\newblock {\em Inf. Process. Lett.}, 38(5):277--278, 1991.

\bibitem{McGarvey1953}
D.~C. McGarvey.
\newblock A theorem on the construction of voting paradoxes.
\newblock {\em Econometrica}, 21(4):608--610, 1953.

\bibitem{Miller1977}
N.~R. Miller.
\newblock Graph-theoretical approaches to the theory of voting.
\newblock {\em AM. J. Polit. Sci.}, 21:769--803, 1977.

\bibitem{DBLP:journals/ai/NevelingR21}
M.~Neveling and J.~Rothe.
\newblock Control complexity in {Borda} elections: Solving all open cases of
  offline control and some cases of online control.
\newblock {\em Artif. Intell.}, 298:Nr.~103508, 2021.

\bibitem{Rasch2000}
B.~E. Rasch.
\newblock Parliamentary floor voting procedures and agenda setting in {Europe}.
\newblock {\em Legislative Studies Quarterly}, 25(1):3--23, 2000.

\bibitem{DBLP:conf/sigecom/SornatWX21}
K.~Sornat, V.~{Vassilevska Williams}, and Y.~Xu.
\newblock Fine-grained complexity and algorithms for the {Schulze} voting
  method.
\newblock In {\em EC}, pages 841--859, 2021.

\bibitem{DBLP:journals/interfaces/Tovey02}
C.~A. Tovey.
\newblock Tutorial on computational complexity.
\newblock {\em Interfaces}, 32(3):30--61, 2002.

\bibitem{DBLP:journals/aarc/Veselova16}
Y.~A. Veselova.
\newblock Computational complexity of manipulation: {A} survey.
\newblock {\em Autom. Remote. Control.}, 77(3):369--388, 2016.

\bibitem{DBLP:journals/amai/Walsh11}
T.~Walsh.
\newblock Is computational complexity a barrier to manipulation?
\newblock {\em Ann. Math. Artif. Intell.}, 62(1-2):7--26, 2011.

\bibitem{DBLP:journals/jair/Walsh11}
T.~Walsh.
\newblock Where are the hard manipulation problems?
\newblock {\em J. Artif. Intell. Res.}, 42:1--29, 2011.

\bibitem{DBLP:conf/aaai/Williams10}
V.~V. Williams.
\newblock Fixing a tournament.
\newblock In {\em AAAI}, pages 895--900, 2010.

\bibitem{williamstournamentinhandbook2016}
V.~V. Williams.
\newblock Knockout tournaments.
\newblock In F.~Brandt, V.~Conitzer, U.~Endriss, J.~Lang, and A.~Procaccia,
  editors, {\em Handbook of Computational Social Choice}, chapter~19, pages
  453--474. Cambridge University Press, New York, USA, 2016.

\bibitem{DBLP:journals/jair/XiaC11}
L.~Xia and V.~Conitzer.
\newblock Determining possible and necessary winners given partial orders.
\newblock {\em J. Artif. Intell. Res.}, 41:25--67, 2011.

\bibitem{DBLP:conf/atal/Yang17}
Y.~Yang.
\newblock The complexity of control and bribery in majority judgment.
\newblock In {\em AAMAS}, pages 1169--1177, 2017.

\bibitem{AAMAS17YangBordaSinlgePeaked}
Y.~Yang.
\newblock On the complexity of {B}orda control in single-peaked elections.
\newblock In {\em AAMAS}, pages 1178--1186, 2017.

\bibitem{DBLP:conf/ecai/000120}
Y.~Yang.
\newblock On the complexity of constructive control under nearly single-peaked
  preferences.
\newblock In {\em ECAI}, pages 243--250, 2020.

\bibitem{DBLP:conf/atal/000123a}
Y.~Yang.
\newblock On the complexity of the two-stage majoritarian rule.
\newblock In {\em AAMAS}, pages 2022--2030, 2023.

\bibitem{DBLP:journals/algorithmica/Yang25}
Yongjie Yang.
\newblock On the parameterized complexity of controlling amendment and
  successive winners.
\newblock {\em Algorithmica}, 87(6):842--907, 2025.

\end{thebibliography}
\end{document}